\documentclass[11pt,letterpaper]{article}
\usepackage[margin=1in]{geometry}

\usepackage[utf8]{inputenc}
\usepackage{amsmath, amsfonts, amssymb, amsthm, dsfont}
\usepackage{babel}
\usepackage{parskip}
\usepackage{enumerate}
\usepackage{enumitem}
\usepackage[normalem]{ulem} 
\usepackage{graphicx}
\usepackage{mathtools}
\usepackage{xspace}
\usepackage{comment}
\usepackage{bm}
\usepackage{soul}
\usepackage{tcolorbox}
\usepackage{aliascnt}

\usepackage[usenames,dvipsnames]{xcolor}
\usepackage[backref=page,linktocpage=true,breaklinks,colorlinks,citecolor=PineGreen,linkcolor=Violet]{hyperref}
\usepackage{cleveref}

\usepackage{algorithm}
\usepackage[noend]{algpseudocode}

\newcommand{\R}{\mathbb{R}}

\newcommand{\E}{\mathbb{E}}

\newcommand{\poly}{\mathrm{poly}}

\newcommand{\Prb}{\mathbb{P}}
\newcommand{\1}{\mathbf{1}}
\newcommand{\eps}{\varepsilon}
\newcommand{\OPT}{\mathrm{OPT}}
\newcommand{\ALG}{\mathrm{ALG}}
\newcommand{\Unif}{\mathrm{Unif}}
\newcommand{\CutReg}{\mathrm{CutReg}}
\newcommand{\calF}{\mathcal F}
\newcommand{\argmax}{\operatorname*{argmax}}

\newtheorem{theorem}{Theorem}[section]
\newaliascnt{lemma}{theorem}
\newtheorem{lemma}[lemma]{Lemma}
\aliascntresetthe{lemma}

\newaliascnt{proposition}{theorem}

\aliascntresetthe{proposition}

\newaliascnt{corollary}{theorem}

\aliascntresetthe{corollary}

\newaliascnt{definition}{theorem}
\newtheorem{definition}[definition]{Definition}
\aliascntresetthe{definition}

\newaliascnt{remark}{theorem}

\aliascntresetthe{remark}

\newaliascnt{observation}{theorem}
\newtheorem{observation}[observation]{Observation}
\aliascntresetthe{observation}

\newcommand{\vertm}[1]{{\left\vert\kern-0.25ex\left\vert\kern-0.25ex\left\vert #1     \right\vert\kern-0.25ex\right\vert\kern-0.25ex\right\vert}}

\newcounter{note}[section]

\title{Improved Combinatorial Auctions with and without Samples 
}

\author{ Aviad Rubinstein\thanks{(aviad@cs.stanford.edu) Department of Computer Science, Stanford University. Supported in part by David and Lucile Packard Fellowship.} \and Sahil Singla\thanks{
        (ssingla@gatech.edu)
        School of Computer Science,
        Georgia Tech.
        Supported in part by NSF award  CCF-2440113.
        }}

\newcommand{\Welfare}{\ensuremath{\mathsf{Welfare}}}

\newcommand{\coef}{\operatorname{coef}_\tau}

\begin{document}

\maketitle

\vspace{-.25cm}
\begin{abstract}
\vspace{-.25cm}

We give new mechanisms for combinatorial auctions in the offline, secretary, and prophet-with-samples models. Specifically, we obtain the following guarantees for social welfare maximization with $m$ items and $n$ XOS buyers:
\smallskip
\begin{itemize}
    \item In the prophet-with-samples version, buyers arrive in a fixed known adversarial order
    and draw their values from independent distributions. We give a universally truthful, polynomial-time \((2+\eps)\)-approximation mechanism  using $O(\log\!\log m)$ samples from each distribution. The approximation factor is tight even for a single item and without the incentive constraints, due to the classic prophet inequality. Compared to the previous state-of-the-art with the same guarantees [D\"utting--Kesselheim--Lucier--Reiffenh\"auser--Singla (FOCS 2024)], our running time is exponentially faster, and our sample complexity is {\em doubly}-exponentially lower, resolving two explicit open problems by D\"utting et al.
    \smallskip
    \item Using the same techniques, we also obtain improved mechanisms in the  secretary and offline versions, where the values are worst-case, and buyers are processed in either random order (secretary) or without any timing restrictions (offline). 
    After a long line of works, the state of the art in both models is an \(O((\log\!\log  m)^2)\)-approximation due to [Assadi--Kesselheim--Singla (SODA 2021)], who explicitly ask whether it can be improved to \(O(\log\!\log  m)\). We indeed obtain an \(O(\log\!\log  m)\)-approximate universally truthful mechanism for the secretary and offline models.
\end{itemize}

\smallskip
The main technical insight is that the median-prices approach of D\"utting et al. can be viewed as a Nash equilibrium of a certain game, but we can get exponentially improved sample complexity by instead computing an approximate correlated equilibrium using swap regret minimization. To obtain the doubly-logarithmic sample bound, we design an application-specific notion of regret that can be minimized even more efficiently.

\end{abstract}

{\footnotesize
\setcounter{tocdepth}{1}
\tableofcontents
}

\clearpage

\section{Introduction}

A fundamental problem in economics and theoretical computer science is to allocate heterogeneous
items to strategic buyers.  There are \(m\) items and \(n\) buyers, and each buyer has a private
valuation function over bundles of items.  The goal is to allocate the items so as to maximize
social welfare, the total value of the allocated bundles.  The difficulty is that the buyers'
valuations are private, high-dimensional objects, and the allocation rule should be
\emph{truthful}: buyers should have no incentive to misreport their values.

The Vickrey-Clarke-Groves mechanism is truthful and welfare-maximizing
\cite{Vickrey61,Clarke71,Groves73}, but implementing it requires solving the welfare maximization
problem exactly and obtaining enough information from the buyers about their valuations.  For
combinatorial valuations this is generally computationally and communicationally intractable;
indeed, understanding the power and limitations of truthful approximation mechanisms for
combinatorial auctions has been a central theme in algorithmic mechanism design since its
inception \cite{NisanR99,Nisan00,NisanS06}.  On the algorithmic side, for large complement-free
classes such as submodular and XOS valuations, constant-factor welfare approximations are known in
standard oracle models \cite{DobzinskiNS05,Vondrak08,Feige09,FeigeV10}.  However, these algorithms
are not, in general, truthful mechanisms:  a randomized mechanism is \emph{universally truthful} if, for every fixed realization of its internal randomness, truthful reporting is a dominant strategy for every buyer, regardless of the other buyers' reports.
The motivating question is whether one can obtain the best of both worlds: simple truthful mechanisms with strong welfare guarantees and polynomial running time.

The simplest such mechanisms are \emph{sequential posted prices}. 
The seller posts one price per item, buyers arrive one by one, and each buyer chooses her favorite remaining bundle at the posted prices. 
Since every buyer faces posted prices independent of her own report, the mechanism is truthful.  Posted prices are also inherently online: buyers need not be simultaneously present, and the allocation is implemented by demand maximization against the currently available set.

\paragraph{Stochastic buyers and samples.}


A beautiful result of Feldman, Gravin, and Lucier~\cite{FeldmanGL15} shows that, when buyers'
valuations are drawn from independent distributions, there exist anonymous item prices that guarantee half of
the expected optimal welfare for XOS buyers.

The main limitation of \cite{FeldmanGL15}, and its generalization in \cite{DuettingFKL17}, is that it requires 
complete knowledge of the buyers' distributions. 
They also give algorithmic versions: in $\poly(m)$ samples 
they can obtain an {\em additive} approximation guarantee. Intuitively, additive approximation is a natural barrier to their approach which is based on mean-based prices that are sensitive to ``tail'' high prices that are hard to estimate from samples.

D\"utting, Kesselheim, Lucier,
Reiffenh\"auser, and Singla~\cite{DKLRS-FOCS24} obtain the first tight multiplicative \((2+\eps)\)-approximation from samples.  Their mechanism is based on a high-dimensional analogue of median prices: roughly, a price vector should make every item sell with probability exactly \(1/2\).  
The existence of such prices is proved through a fixed-point argument, and their truthful posted-price mechanism uses polynomially many samples.  
Their work leaves two explicit questions relevant here: can we learn a truthful constant-factor posted-price mechanism from polylogarithmically many samples, and whether the XOS construction can be implemented in polynomial time?

\paragraph{Beyond distributional information.}
A second line of work studies truthful combinatorial auctions without any distributional prior over the buyers.  Here the goal is to give a truthful approximation for every fixed valuation profile. This prior-free setting has been studied extensively, both offline and online, through
maximal-in-range and randomized mechanisms, convex-rounding based mechanisms, and posted-price or random-order mechanisms \cite{DobzinskiNS06,Dobzinski07,LaviS05,DughmiR10,DughmiRY11,Dobzinski16,KrystaV12}.

A particularly robust way to make such guarantees online is the random-order model: the valuation profile is fixed adversarially, but buyers arrive in a uniformly random order.  Any such mechanism also gives an offline truthful mechanism, since an offline mechanism can randomly order the buyers and simulate the online rule; the advantage is that the same guarantee remains valid when buyers are not simultaneously present. 

In the offline model, the approximation ratio has improved over a long line of work spanning over two decades, starting from $O(\sqrt{m})$~\cite{DobzinskiNS05}%
\footnote{This is an offline, deterministic mechanism (as opposed to random-order arrival). For deterministic mechanisms, $O(\sqrt{m})$-approximation was the state of the art for nearly two decades until it was recently improved to $O(\sqrt{m/\log m})$~\cite{QW-STOC24}.}, followed by randomized guarantees of $O(\log^2 m)$~\cite{DobzinskiNS06}, $O(\log m \log\!\log m)$~\cite{Dobzinski07}, $O(\log m)$~\cite{KrystaV12}, $O(\sqrt{\log m})$~\cite{Dobzinski16}, $O((\log\!\log m)^3)$~\cite{AS-FOCS19}, and finally $O((\log\!\log m)^2)$~\cite{AKS-SODA21}. 
The latter work asks whether the approximation ratio can be improved to $O(\log\!\log m)$.


\subsection{Our Results}

We state the results here; the precise model and technical overview appear in \Cref{sec:overview}.  Throughout, $\OPT$ denotes the optimal welfare.  In the stochastic results, the expectation is over the valuation profile and the mechanism's randomness. In the random-order result, the valuation profile is fixed and the expectation is over the random arrival order and the mechanism's randomness.  A sample means one independent valuation profile drawn from the product
distribution.

Our main theorem is a sample-efficient version of  Feldman-Gravin-Lucier posted-price guarantee.

\begin{theorem}[Main theorem]\label{thm:intro-main}
For independent XOS buyers in any fixed known arrival order and every \(\eps\in(0,1)\), there is a
polynomial-time universally truthful randomized item posted-price mechanism using
\[
  O\!\left(\eps^{-2}\log\frac{\log m}{\eps}\right)
\]
samples and satisfying
  $\E[\ALG]\ge (1/2-\eps)\E[\OPT]$.
For constant \(\eps\), the mechanism is \((2+O(\eps))\)-competitive with
\(O(\log\!\log  m)\) sample complexity.
\end{theorem}

Thus, for constant accuracy, the sample complexity of truthful posted pricing drops from polynomial in the market size~\cite{DKLRS-FOCS24} to doubly logarithmic in the number of items. 
 

The same ideas 
also give a prior-free random-order guarantee.  Here the valuation profile is fixed adversarially, and the only randomness comes from the arrival order and the mechanism.

\begin{theorem}[Random order/offline]\label{cor:intro-random-order}
For every fixed XOS valuation profile, there is a universally truthful online mechanism in
uniformly random arrival order, mixing an item posted-price branch with a grand-bundle posted-price
branch, whose expected welfare is
  $\Omega\!\left(\frac{\OPT}{\log\!\log  m}\right)$.
\end{theorem}

In the  XOS setting, this improves the \(O((\log\!\log  m)^2)\) guarantee obtained from the Assadi-Singla and Assadi-Kesselheim-Singla line~\cite{AS-FOCS19,AKS-SODA21}. The proof relies on the posted-price structure developed for the main theorem.

Finally, the same ideas also yield a result at the extreme low-sample end of the stochastic model, where we only have a single sample from each buyer's distribution.  With only one sample, we cannot run the full distributional training procedure of \Cref{thm:intro-main}; instead, the sample creates sparse training roles and estimates the relevant scale, while the real buyers are served only in the allocation role.  

\begin{theorem}[Single sample]\label{cor:intro-one-sample}
For independent XOS buyers in any fixed known arrival order, one sample suffices for a universally
truthful randomized posted-price mechanism, mixing an item posted-price branch with a grand-bundle
posted-price branch, with
  $\E[\ALG]\ge \Omega\!\left(\frac{\E[\OPT]}{\log\!\log  m}\right)$.
\end{theorem} 

 Theorems~\ref{cor:intro-random-order} and~\ref{cor:intro-one-sample} are proved together in \Cref{sec:SSRO}, where the new ingredient is again the posted-price structure developed for the main theorem.
 

\subsection{Further Related Work}

\noindent \textbf{Truthful combinatorial auctions.}
Truthful combinatorial auctions have been studied extensively since the beginning of algorithmic mechanism design~\cite{NisanR99,Nisan00}.  For restricted valuation classes, positive approximation results include \cite{LehmannLN02,LehmannLN06,DobzinskiNS05,DobzinskiNS06,LaviS05,Dobzinski07,Dobzinski16}, and convex-rounding and black-box reduction approaches were developed in \cite{DughmiR10,DughmiRY11}.  There are also strong computational and communication barriers for truthful mechanisms~\cite{NisanS06,PapadimitriouSS08,Dobzinski11,DughmiV11,DobzinskiV12,DobzinskiV13,DobzinskiNO14,EzraFNTW19,AKSW-STOC20}.  Our work follows the line of simple truthful mechanisms via posted prices in online
combinatorial auctions.
 
\smallskip
\noindent \textbf{Posted prices, prophet inequalities, and samples.}
Posted pricing is a central tool in prophet inequalities and mechanism design
\cite{HajiaghayiKS07,ChawlaHMS10}.  The closest known-distribution benchmark for us is the
posted-price mechanism of Feldman, Gravin, and Lucier~\cite{FeldmanGL15}; the broader balanced-price
framework was developed in~\cite{DuettingFKL17}.  The single-item median threshold of
Samuel-Cahn~\cite{SamuelCahn84} is the classical sample-robust prototype.  Prophet inequalities
with limited or sampled information have since been studied extensively, including
\cite{AzarKW14,CorreaDFS19,RubinsteinWW20,CorreaDFSZ21,CorreaCES20,KaplanNR20,CristiZilliotto24}.

\smallskip
\noindent \textbf{Random order and prior-free online mechanisms.}
The random-order model originates in online matching and secretary-type problems; e.g., see book chapter \cite{GS-Book20}.  For online combinatorial auctions, \cite{KrystaV12,KRTV-ESA13,Dobzinski16} developed online mechanism-design tools, and Assadi-Singla and
Assadi-Kesselheim-Singla~\cite{AS-FOCS19,AKS-SODA21} obtained improved truthful mechanisms for
submodular and subadditive buyers.  

\smallskip
\noindent \textbf{Weaker notions of calibration.}
Our cut-calibration condition is conceptually related to step calibration~\cite{QZ-COLT25}, cutoff calibration~\cite{RSBRW-COLT25}, and proper calibration~\cite{OKK-FOCS25}. These notions replace levelwise absolute calibration by controlling signed errors against structured threshold or interval tests, retaining decision-theoretic guarantees while reducing statistical or computational cost. Our setting differs: the ordered levels are item prices, the target sale probability is $1/2$ at every level, and the welfare proof requires only one orientation of each proper cut. We show that this auction-specific condition suffices for the welfare guarantee and can be learned from smoothed posted-price feedback.

\section{Model and Technical Overview} \label{sec:overview}

\subsection{Model}\label{sec:model}
Let \(M\) denote the set of \(m:=|M|\) items. There are \(n\) buyers with valuation functions \(v_1,\ldots,v_n:2^M\to\R_+\).
We will assume that all valuations $v$ are normalized (i.e., $v(\emptyset)=0$) and monotone.  A valuation $v$ is XOS~\cite{LehmannLN06} if there is a collection of additive functions $a^1,a^2,\ldots$ such that
\[
  v(S)=\max_\ell\sum_{j\in S}a^\ell_j
  \qquad\text{for every }S\subseteq M.
\]
For normalized monotone XOS valuations, the support coefficients $a^\ell_j$ can be taken nonnegative.

A sequential item \emph{posted-price mechanism} fixes prices $p_j\ge0$ for $j \in M$ before seeing the buyer.  When buyer $i \in [n]$ with valuation $v_i$ arrives and the remaining set is $R \subseteq M$, she receives 
\[
  \argmax_{S\subseteq R}\Big\{v_i(S)-\sum_{j\in S}p_j \Big\}.
\]
We use a fixed global tie-breaking rule over bundles, independent of reports and prices. 

For the stochastic results, buyer valuations $v_i$ are drawn independently from unknown distributions $D_i$ over normalized monotone XOS valuations. Let $D:=D_1\times\cdots\times D_n$, and index the buyers in their fixed known arrival order. For any valuation profile $v=(v_1,\ldots,v_n)$, let $\OPT(v)$ denote its optimum welfare.

Fix a rule choosing an optimal allocation $O(v)=(O_1(v),\ldots,O_n(v))$ and, for each buyer $i$, a supporting XOS clause $a_i(v)=(a_{i,j}(v))_{j\in M}$ for $O_i(v)$. Thus
\[
  v_i(S)\ge\sum_{j\in S}a_{i,j}(v)\quad\text{for every }S\subseteq M,
  \qquad
  v_i(O_i(v))=\sum_{j\in O_i(v)}a_{i,j}(v).
\]
Define $w_j(v)= a_{i,j}(v)$ if $j\in O_i(v)$ and $w_j(v)= 0$ if $j$ is unallocated in $O(v)$. 
The bundles $O_i(v)$ are disjoint, so this definition is unambiguous and $\sum_jw_j(v)=\OPT(v)$. Define $\mu_j:=\E_{v\sim D}[w_j(v)]$; then $\sum_j\mu_j=\E[\OPT]$. These benchmark objects are used only in the analysis: the mechanism does not compute an optimal allocation or its supporting clauses.

In the stochastic results, a \emph{fresh valuation profile} means the independent draw $v\sim D$ on which the mechanism is evaluated, as opposed to the training samples. It is independent of the training samples and of the randomness used to construct the posted prices.

\paragraph{Oracle and sample access.}
A sample means one independent valuation profile
\((\hat v_1,\ldots,\hat v_n)\sim D_1\times\cdots\times D_n\).  Thus \(T\) samples means \(T\) independent samples from each buyer distribution.  For sampled valuations we assume oracle access sufficient to make value queries and (tie-broken) demand queries to simulate posted-price executions on sampled profiles. 
For real buyers, the allocation phase is implemented as a posted-price mechanism: after the mechanism's internal randomness and all training samples are fixed, every buyer faces prices that are independent of her own report.  Hence every deterministic realization is DSIC, and the randomized mechanisms are universally truthful.

\paragraph{The mechanism at a glance.}
Our main mechanism has a training phase on sampled valuation profiles and an allocation phase on the real buyers. The complete specification is in Algorithm~\ref{alg:main-training}; its three stages are as follows.
\begin{enumerate}[leftmargin=2em]
  \item \emph{Preprocessing:} Use preprocessing samples to choose a price window and a finite geometric grid of base prices. (If the preprocessing scale is zero, we handle it separately in \Cref{lem:window}.)
  \item \emph{Training:} Run $T$ training rounds, with one learner per item. Before round $t$'s sample is drawn, learner $j$ outputs and stores a distribution $\pi_{j,t}$ over the grid. Draw one independent valuation-profile sample. For each item $j$, draw the other items' prices from their current learner distributions and apply independent perturbations. Holding those other prices fixed, simulate the posted-price mechanism on the sample while varying only $j$'s price. The resulting smoothed sale probabilities at all grid levels are the feedback used to update learner $j$. The same valuation-profile sample supplies feedback for every item.
  \item \emph{Real execution:} Choose one common round $U$ uniformly from $[T]$. Independently for each item $j$, draw its base price from $\pi_{j,U}$ and apply a fresh perturbation. Fix the resulting price vector and offer it sequentially to the real buyers in their fixed known order.
\end{enumerate}
Training simulations do not allocate real items, and prices are not updated in response to the real buyers' reports. The rounds below refer to training rounds, not to individual buyer arrivals. We next explain which property of the learned price distribution guarantees high welfare and how the learners obtain it.

\subsection{Technical Overview}
 In the Feldman, Gravin, and Lucier (FGL) mechanism \cite{FeldmanGL15}, the price of an item is half times its expected contribution to the offline optimum. More generally, the balanced-price framework \cite{DuettingFKL17} also sets prices from expected marginal values of an appropriate benchmark allocation.  These prices are extremely powerful when the distributions are known. However, they are much less suitable as a few-sample learning target: for unbounded distributions, an item's expected contribution may be dominated by
rare events that are invisible in any small number of samples.

\paragraph{Median approach of \cite{DKLRS-FOCS24}.}
 In the single-item prophet inequality, the threshold median of the maximum value~\cite{SamuelCahn84} is robust to tails, and hence learnable, because it depends on a quantile rather than a mean. The authors of \cite{DKLRS-FOCS24} imported this idea into combinatorial auctions by defining a high-dimensional analogue of median prices: roughly,
a price vector is a combinatorial median if every item sells with probability exactly $1/2$.  This
condition is strong enough to recover the factor-two FGL welfare accounting, and \cite{DKLRS-FOCS24}
show that such prices exist and can be learned from polynomially many samples.
One difficulty is that high-dimensional medians are not canonical.  Unlike in one dimension, there
is no unique way to define ``the'' median price vector.  Roughly, the proof of existence of median prices in \cite{DKLRS-FOCS24} can be viewed as existence of \emph{Nash equilibrium} of a certain game: consider a game where items decrease or increase their prices depending on whether their probability of sale is below or above $1/2$, respectively. The Nash equilibrium of this game exists due to a fixed point theorem. However, this is an existence proof and does not give a samples/time efficient algorithm. 

Our main contribution is to identify a weaker median condition: instead of learning Nash equilibrium, we learn the analogue of an \emph{approximate correlated equilibrium} for a restricted class of local deviations.  This weaker object still suffices for the welfare proof and can be learned much faster.
Now our approach consists of three main ideas that we describe next.

\paragraph{Idea 1: calibrated median prices.}
We  relax the goal from finding one deterministic median price vector to finding a randomized distribution over price vectors.  
Assume that we are given a price window \([L,B]\) such that 
\[
    B/L=\poly(m,1/\eps) \qquad \Prb_{v\sim D}[\OPT(v)\ge B]\le 1/2,
    \qquad \text{and } \quad
    mL\le O(\eps \E[\OPT(v)])
    .
  \]
 This window is easily learnable from a few samples; see \Cref{lem:window}. 
Now fix an item \(j\) and a geometric grid (like powers of $(1+\Theta(\epsilon^2))$) of possible base
prices%
\footnote{We call them ``base'' prices because, as explained in Idea 2, we later slightly perturb those base prices to obtain a smoothness property useful for our reduction to an online learning problem.}
\(p_1,\ldots,p_N\) in $[L,B]$.

To simplify notation, we sometimes call price $p_r$ the price at \emph{level} $r$. Let $X=(X_1,\ldots,X_m)$ be a random base-price vector on this grid, and define
\[
  \lambda_j(r)=\Pr[X_j=p_r],
  \qquad
  q_j(r)=\Pr[j\text{ sells}\mid X_j=p_r].
\]
Here the experiment is to draw $X$, apply the independent price perturbations described in Idea 2, independently draw a valuation profile $v\sim D$, and run the sequential posted-price mechanism of \Cref{sec:model} in the fixed buyer order. ``Item $j$ sells'' means that $j$ is allocated by the end of this execution. The probability averages over the base prices, the perturbations, and $v$; the tie-breaking rule is fixed. If $\lambda_j(r)=0$, set $q_j(r)=0$, since it only appears multiplied by $\lambda_j(r)$.

A natural randomized median condition, inspired by the no-profitable-deviation viewpoint behind approximate correlated equilibrium, is\footnote{For intuition, we suppress boundary corrections: at the lowest
price, only selling too often is penalized, and at the highest price,
only selling too rarely is penalized. These one-sided corrections are
handled in  \Cref{thm:cut-median}.}
\[
  \sum_{r=1}^N \lambda_j(r)\left|q_j(r)-\frac12\right|\le O(\eps)
  \qquad\text{for every item }j .
\]
This is an \(\ell_1\)-calibration condition: conditioned on the price level actually used by the
mechanism, the item sells with probability close to \(1/2\) on average.

This condition is already enough for a \((2+\eps)\)-approximation.  The reason is that the standard
utility/revenue certificate \cite{FeldmanGL15,DKLRS-FOCS24} for an item has the same structure as the one-dimensional median
argument:  If an item sells, we get revenue; if it does not sell, its owner in the optimum must be making up the utility by taking other items.  When \(q_j(r)\) is close to \(1/2\) on the prices used by the mechanism,
these two terms balance in exactly the way needed for the factor-two welfare proof.  Thus exact combinatorial medians are not necessary; approximate randomized medians suffice.

\paragraph{Idea 2: learning calibrated medians by smoothed swap regret.}
The next question is how to learn such calibrated median prices.  Standard preprocessing reduces
the relevant price range to a geometric grid with
  $N=\widetilde O_\eps(\log m)$
levels, where recall that $m$ is the number of items.  For a fixed sample profile and fixed prices of all other items, the sold indicator of item
\(j\) is monotone in its own price.  Hence there is a critical price: below it \(j\) sells, and above
it \(j\) does not sell.  This gives full-information feedback for item \(j\)'s entire one-dimensional
price-learning problem from a single sample profile.

There is one technical obstacle.  If the critical price lies between two
adjacent grid points, say \(p_r<\Theta<p_{r+1}\), then the item sells at level \(r\) but not at
level \(r+1\), so the adjacent sale indicators differ by \(1\).  To make local regret meaningful,
we therefore do not post the base price directly.  Instead, after choosing a base price level \(r\),
the mechanism multiplies \(p_r\) by an independent log-uniform perturbation
in a short interval \([1,e^\tau]\), where \(\tau=\Theta(\eps)\).  
This smoothing turns the critical price into a Lipschitz sale curve along the grid: adjacent levels have sale probabilities that differ by only \(O(1/\tau \cdot \log(1+\beta)) = O(\eps)\), where $1+\beta$ is the grid ratio for $\beta= \Theta(\eps^2)$.

With this smoothing, one can reduce the boundary-aware calibration condition to a $\Phi$-regret online learning problem. In each round the learner chooses a distribution over price levels. The training environment induces a random critical level $C\in\{0,\ldots,N\}$ under the own-price perturbation, and the scalar loss at level $r$ is $\E[|r-C|]$; see \Cref{sec:sale-curves}. For this warm-up, $\Phi$ consists of all maps satisfying $\phi(r)\in\{r-1,r,r+1\}\cap[N]$ for every $r$. Every one-round deviation gain lies in $[-1,1]$, and $\log|\Phi|=O(N)$. Thus external-regret-to-$\Phi$-regret reductions~\cite{BM-JMLR07,GGM-ICML08} give $O(\sqrt{TN})$ regret and, using smoothing, $\widetilde O_\eps(N)=\widetilde O_\eps(\log m)$ samples for this calibration condition. This does not yet meet our ultimate goal of $O(\log\!\log m)$ samples, but it already shows that the polynomial sample dependence of~\cite{DKLRS-FOCS24} is not inherent.


We remark that one natural approach to improve the learning rate is to apply recent large-action-space swap-regret algorithms~\cite{PengR24,DaganDFG24}. Unfortunately, they do not seem to apply because even though the $\Phi$-regret comparing to adjacent price levels is $\le 1$, the loss  ($|r-b|$) at price level \(r\) in 
any round can be as large as $N$, and the algorithms of~\cite{PengR24,DaganDFG24}  have a linear dependence on the range of loss. 

\paragraph{Idea 3: cut-calibrated median prices.} To improve the learning rate from $O(N)$ to $O(\log N)$, the key observation is that the welfare proof does not require arbitrary levelwise corrections. As we explain below, sale-probability errors enter the welfare certificate with nondecreasing coefficients along the price grid. After handling the endpoints, every such coefficient sequence is a convex combination of single-cut step functions. Thus it suffices to compete with the $N-1$ cut deviations, rather than the exponentially many local deviation patterns considered
in Idea 2. 

To formalize this weaker calibration requirement, 
for item \(j\), we consider signed mass
  $\lambda_j(r)\left(q_j(r)-\frac12\right)$
on level \(r\), instead of unsigned mass $\lambda_j(r)|q_j(r)-\frac12|$.  Positive mass means that the level sells too often; negative mass means that it
sells too rarely. 
For each proper cut \(b\in\{1,\ldots,N-1\}\) of the price line  (corresponding to a cutoff price in the auction setting), consider the signed imbalance
\[
  \sum_{r\le b}\lambda_j(r)\left(q_j(r)-\frac12\right)
  -
  \sum_{r>b}\lambda_j(r)\left(q_j(r)-\frac12\right).
\]
We define our calibration error \(\Delta_j\) to be the positive part of the maximum of this quantity over all proper cuts.  

The proof that this weaker condition suffices is the main welfare lemma.  
For one item, after normalizing by its expected contribution, the FGL certificate is lower bounded by a
one-dimensional expression of the form
  $(1-q)(1-e^\tau z)_+ + zq $,
where \(q\) is the sale probability and \(z\) is the normalized price scale.
This function is bounded below by
  $\frac12 + \theta(z)\left(q-\frac12\right)-O(\tau)$,
where \(\theta(z)\) is a nondecreasing coefficient going from \(-1\) at low prices to \(+1\) at high
prices.  
Therefore sale-probability errors appear with monotone coefficients along the price grid in the welfare bound.  Since any such monotone coefficient sequence is a convex combination of cut step functions, the total signed loss is controlled by the maximum cut imbalance
\(\Delta_j\).
This gives the cut-calibrated median welfare bound
\[
  \E[\ALG]
  \ge
  \left(\frac12-O(\tau)\right) \E[\OPT]
  -
  \sum_j \mu_j\Delta_j
  -
  e^\tau Lm ,
\]
where \(\mu_j\) is item \(j\)'s expected contribution to the optimum and \(L\) is the bottom of the
price grid.  

Finally, the advantage of cut calibration is that it can be learned much faster than full swap calibration.  Since we compete only with the \(N-1\) cuts, a stationary-distribution reduction from external regret to swap regret gives a learner whose cut regret is \(O(\sqrt{T\log N})\). 
A standard martingale argument then transfers the empirical cut calibration learned on samples to the fresh valuation profile faced by the final posted-price mechanism. 
Taking   $T=O(\eps^{-2}\log N)$
therefore makes \(\E[\Delta_j]=O(\eps)\) for every item \(j\).
Importantly, we do not need a union bound over items: one sampled profile gives feedback for all items, and the welfare loss is the weighted sum \(\sum_j\mu_j\Delta_j\), so linearity of expectation is enough.

\paragraph{The one-sample and random-order consequences.} Items with very small contributions are negligible because the grid floor is low.  Items with very
large contributions are charged to a grand-bundle branch: in the one-sample model we use a single-sample grand-bundle prophet inequality, and in the random-order model we use the classical secretary threshold rule for the grand bundle. In the remaining cases, we first obtain a scale \(\Gamma\), then run the same cut-calibrated learner on \(H=O(\log N)=O(\log\!\log  m)\) training blocks, and finally post a uniformly random stored price vector to the allocation block.
  Since \(H=O(\log N)=O(\log\!\log  m)\), this yields the
one-sample and random-order \(O(\log\!\log  m)\) guarantees.

\section{Cut-Calibrated Median Prices} \label{sec:cutCalib}

We introduce the main cut-calibration property of our  distribution over price vectors. We will first introduce the cut-calibration error and then show that small cut-calibration error implies good welfare. In \Cref{sec:learning} we will show how to learn such a distribution over price vectors with a small number of samples.

Assume that we are given a price window \([L,B]\) such that 
\[
    \Prb_{v\sim D}[\OPT(v)\ge B]\le 1/2,
    \qquad
    mL\le O(\eps \cdot \E[\OPT(v)]),
    \quad \text{and } \quad 
    B/L=\poly(m,1/\eps).
  \]
 This window is easily learnable from a few samples (see \Cref{lem:window}).

\subsection{Base Price Grid, Smoothing, and Cut-Calibration Error}
\label{sec:priceGrid}

Fix a geometric grid of prices:
\[
  p_1<p_2<\cdots<p_N,
  \qquad p_{r+1}=(1+\beta)p_r,
\]
such that \(p_1\le L\) and \(p_N\ge B\). 
(We will eventually set $\beta = \Theta(\epsilon^2)$ and $N=O(\epsilon^{-2} \log m)$.)  

The \emph{base price} of item $j$ is a random grid value $X_j\in\{p_1,\ldots,p_N\}$.  The actual posted price $P_j$ is obtained by \emph{smoothing} base price $X_j$:
\[
  P_j=X_jR_j
  \quad \text{with } \quad R_j=\exp(U_j) \text{ for } U_j\sim\Unif[0,\tau],
\]
where perturbations are independent across items and independent of the fresh valuation profile (we will eventually set $\tau=\Theta(\epsilon)$).  The base-price vector $X=(X_1,\ldots,X_m)$ is also independent of the fresh valuation profile, although the coordinates of $X$ may be arbitrarily correlated.

For item $j$ and level $r$, define
\begin{align} \label{eq:defnLambdaQ}
  \lambda_j(r):=\Prb[X_j=p_r]
  \quad \text{ and } \quad
  q_j(r):=\Prb[j\text{ sells}\mid X_j=p_r].
\end{align}
Here $j$ sells if it is allocated by the end of the sequential posted-price execution on $v\sim D$ in the fixed buyer order. Conditional on $X_j=p_r$, the probability averages over the other base prices, all smoothing perturbations, and the independent profile $v$; the tie-breaking rule is fixed. If $\lambda_j(r)=0$, set $q_j(r)=0$.

For any item $j\in M$ and base price $p_r$, it is useful to imagine the signed error
  $\lambda_j(r)\left(q_j(r)-\frac12\right)$
placed on the price line. 
Positive error means that level $r$ sells too often; negative error means that it sells too rarely.
A \emph{cut} $b\in [N]$ detects imbalance: the total signed error to the left of the cut, minus the total signed error to the right.
We are interested in bounding the largest cut error.

\begin{definition}[Cut-Calibration error]\label{defn:Delta-def}
For any joint distribution of the base-price vector $X$, with $\lambda_j$ and $q_j$ defined by \eqref{eq:defnLambdaQ}, the \emph{cut-calibration error} for item $j$ is
\[
  \Delta_j
  :=
  \left[
  \max_{1\le b<N}
  \left(
    \sum_{r\le b}\lambda_j(r)\left(q_j(r)-\frac12\right)
    -
    \sum_{r>b}\lambda_j(r)\left(q_j(r)-\frac12\right)
  \right)
  \right]_+ .
\]
\end{definition}

\subsection{Welfare of Cut-Calibrated Median Prices}

We now show that cut calibration is enough for the FGL-style utility-plus-revenue accounting.  The main new observation is that sale-probability errors are signed: the bad pattern is low prices selling too often and high prices selling too rarely.  This is
exactly what the cut error \(\Delta_j\) measures.

\begin{theorem}[Cut-calibrated median welfare]\label{thm:cut-median}
Assume \(\tau\le1\), \(p_1\le L\), \(p_N\ge B\), and
  $\Prb[\OPT(v)\ge B]\le \frac12$.
For any joint distribution over base prices that is independent of the valuation profile $v\sim D$ on which the mechanism is evaluated, and any independent perturbations supported on $[1,e^\tau]$, the resulting posted-price mechanism satisfies
\[
  \E[\ALG]
  \ge
  \left(\frac12-O(\tau)\right)\E[\OPT]
  -
  \sum_j\mu_j\Delta_j
  -
  e^\tau Lm .
\]
\end{theorem}

The following one-dimensional inequality in \Cref{lem:signed-one-level} is the calculation behind the theorem.  Think of
\(z\) as a normalized price, $e^\tau$ as price perturbation, and \(q\) as the probability that the item sells at that price.  The term
  $(1-q)(1-e^\tau z)_+$
is the utility certificate from the item not selling and  \(zq\) is the revenue certificate from the item selling.   \Cref{lem:signed-one-level}  says that this normalized welfare certificate is \(1/2\) plus a signed sale-probability error whose coefficient increases with the price $z$.

For \(z\ge0\) and \(q\in[0,1]\), define
\[
  f_\tau(z,q):=(1-q)(1-e^\tau z)_+ + zq,
  \qquad
  \coef(z):=\min\{(e^\tau+1)z-1,1\}.
\]

\begin{lemma}[Signed one-level certificate]\label{lem:signed-one-level}
For every \(z\ge0\) and \(q\in[0,1]\), the function
\(\coef(z)\) lies in \([-1,1]\), is nondecreasing in \(z\), and satisfies
  \[
  f_\tau(z,q)
  \ge
  \frac12+\coef(z)\left(q-\frac12\right)-O(\tau).\]
Moreover, if \(q\le1/2\), then
  $f_\tau(z,q)
  \ge
  q-O(\tau)
  =
  \frac12+\left(q-\frac12\right)-O(\tau)$.
\end{lemma}


The proof of this lemma involves some case analysis and is deferred to \Cref{sec:missingSignedOneLevel}. We also record the elementary identity used to pass from monotone coefficients to cut errors.

\begin{observation}[Monotone coefficients are mixtures of cuts]
\label{obseveration:monotone-cut-mixture}
Let $N\ge2$ and $-1=\theta_1\le\cdots\le\theta_N=1$. For each $b\in\{1,\ldots,N-1\}$, define
\[
  \theta^b_r:=
  \begin{cases}
    -1, & r\le b,\\
    +1, & r>b,
  \end{cases}
  \qquad
  \omega_b:=\frac{\theta_{b+1}-\theta_b}{2}.
\]
Then $\omega_b\ge0$, $\sum_{b=1}^{N-1}\omega_b=1$, and
$\theta_r=\sum_{b=1}^{N-1}\omega_b\theta^b_r$ for every $r$.
\end{observation}

\begin{proof}
Nonnegativity follows from monotonicity, and telescoping gives
$\sum_b\omega_b=(\theta_N-\theta_1)/2=1$. For every $r$,
\[
  \sum_b\omega_b\theta^b_r
  ~=~\sum_{b<r}\omega_b-\sum_{b\ge r}\omega_b
  ~=~2\sum_{b<r}\omega_b-1
  ~=~\theta_r-\theta_1-1
  ~=~\theta_r.    \qedhere
\]
\end{proof}

We now complete the proof of \Cref{thm:cut-median}.

\begin{proof}[Proof of \Cref{thm:cut-median}]
We prove an item-wise lower bound and then sum over items.

Fix an item $j$. For each real buyer $i$, draw an independent copy
$\widetilde v_{-i}\sim\prod_{h\ne i}D_h$ of the other buyers'
valuations, independently of the real profile and all price
randomness. Apply the fixed optimal-allocation and supporting-clause
rules from \Cref{sec:model} to $(v_i,\widetilde v_{-i})$. Let $O_i$
be buyer $i$'s bundle in this allocation and let
$a_i=(a_{i,j})_{j\in M}$ be its supporting XOS clause; $a_{i,j}$
is its coefficient for item $j$. These independent copies are used
only in the analysis, not by the mechanism. Define
\[
  s_j:=\sum_i\E[a_{i,j}\1\{j\in O_i\}],
  \qquad
  y_j:=\sum_i\Prb[j\in O_i].
\]
For each $i$, the profile $(v_i,\widetilde v_{-i})$ has distribution
$D$, so $s_j=\mu_j$. Also $\sum_js_j=\E[\OPT]$ and
\[
  y_j ~=~\E_{v\sim D}\!\left[
    \sum_i\1\{j\in O_i(v)\}
  \right] ~\le~ 1.
\]

\paragraph{FGL certificate for welfare.}
Condition on the base price \(X_j=x\) and write \(q=q_j(x)\), where $q_j(x)$ is as defined in \eqref{eq:defnLambdaQ}.  The revenue from item
\(j\) is at least \(xq\), because the actual price is always at least the base price.

We next lower-bound the utility contribution certified by item \(j\).  Condition on the
entire realized price vector \(P\).  Let \(q(P)\) be the probability, over the fresh
valuation profile, that item \(j\) sells under \(P\), and let \(H_i\) be the event that
\(j\) is available when buyer \(i\) arrives.  Conditional on \(P\), the event \(H_i\)
depends only on earlier real buyers, while \((O_i,a_i)\) depends on \(v_i\) and an
independent copy of the other buyers' valuations.  Hence \(H_i\) is independent of
\((a_{i,j},\1\{j\in O_i\})\) conditional on \(P\).  Moreover,
\[
  \Prb[H_i\mid P]\ge 1-q(P),
\]
since if \(j\) remains unsold at the end, then it was available to every buyer.

If \(j\) is available, buyer \(i\) can buy the available subset of \(O_i\) consisting of
items with positive supported surplus.  Since \(P_j\le e^\tau x\), the item-\(j\) part of
this utility is at least
\[
  (a_{i,j}-e^\tau x)_+\1\{j\in O_i\}.
\]
Averaging over \(P\) conditional on \(X_j=x\), and using
\(q=\E[q(P)\mid X_j=x]\), the utility certified by item \(j\) is at least
\[
  (1-q)
  \sum_i\E[(a_{i,j}-e^\tau x)_+\1\{j\in O_i\}]
  \ge
  (1-q)(s_j-e^\tau xy_j)_+ .
\]
Thus, conditional on \(X_j=x\), item \(j\)'s utility-plus-revenue certificate is at least
\begin{equation}\label{eq:item-certificate}
  F_j(x,q):=(1-q)(s_j-e^\tau xy_j)_+ + xq .
\end{equation}

\paragraph{Applying \Cref{lem:signed-one-level} for monotone coefficients.}
If \(s_j=0\), the item-wise lower bound below is trivial, since the right-hand side is at
most the nonpositive bottom-loss term.  Hence assume \(s_j>0\).  For a grid level \(r\),
write
 $ z_r:=\frac{p_ry_j}{s_j}.$
Because \(y_j\le1\), we have \(p_r/s_j\ge z_r\).  Therefore, for every interior level
\(1<r<N\),
\[
  F_j(p_r,q_j(r))
  \ge
  s_j f_\tau(z_r,q_j(r)).
\]
By \Cref{lem:signed-one-level},
\[
  F_j(p_r,q_j(r))
  \ge
  s_j\left(
    \frac12-O(\tau)
    +
    \coef(z_r)\left(q_j(r)-\frac12\right)
  \right),
\]
where coefficient \(\coef(z_r)\) is nondecreasing in \(r\).

At the bottom level, we use only that \(p_1\le L\).  For every \(q\),
\[
  F_j(p_1,q)
  ~\ge~
  (1-q)(s_j-e^\tau p_1y_j) ~\ge~
  s_j(1-q)-e^\tau Ly_j  
  ~=~
  s_j\left(\frac12-\left(q-\frac12\right)\right)-e^\tau Ly_j .
\]
Thus the bottom level has coefficient \(-1\), up to the additive loss \(e^\tau Ly_j\).

At the top level, a sale implies \(\OPT\ge B\).  Indeed, if item \(j\) sells when
\(X_j=p_N\), then its actual price is at least \(p_N\ge B\); the buyer purchases a bundle
with nonnegative utility, so the realized optimum is at least the price paid for that
bundle, and hence at least \(B\).  Since the prices are independent of the fresh valuation
profile,
  $q_j(N)\le \Prb[\OPT\ge B]\le\frac12 $.
The second part of \Cref{lem:signed-one-level} then gives
\[
  F_j(p_N,q_j(N))
  \ge
  s_j\left(
    \frac12+\left(q_j(N)-\frac12\right)-O(\tau)
  \right).
\]
Thus the top level has coefficient \(+1\).

Combining the interior and boundary bounds, there are coefficients
  $-1=\theta_1\le\theta_2\le\cdots\le\theta_N=1$
such that
\begin{equation}\label{eq:item-level-bound}
  F_j(p_r,q_j(r))
  \ge
  s_j\left(
    \frac12-O(\tau)+\theta_r\left(q_j(r)-\frac12\right)
  \right)
  -
  e^\tau Ly_j\1\{r=1\}.
\end{equation}

\paragraph{From monotone coefficients to cut-calibration.} 
Average \eqref{eq:item-level-bound} over the base level \(r\) with weights
\(\lambda_j(r)\), and set
  $e_r:=\lambda_j(r)\left(q_j(r)-\frac12\right)$.
By \Cref{obseveration:monotone-cut-mixture}, write $\theta=\sum_b\omega_b\theta^b$ with $\omega_b\ge0$ and $\sum_b\omega_b=1$. Then
\[
  \sum_r\theta_re_r
  ~=~\sum_b\omega_b\left(-\sum_{r\le b}e_r+\sum_{r>b}e_r\right)
  ~\ge~ -\sum_b\omega_b\Delta_j
  ~=~ -\Delta_j,
\]
where the inequality is by the definition of $\Delta_j$.


Therefore the expected certificate from item \(j\) is at least
\[
  \left(\frac12-O(\tau)-\Delta_j\right)s_j-e^\tau Ly_j
  =
  \left(\frac12-O(\tau)-\Delta_j\right)\mu_j-e^\tau Ly_j .
\]
Summing over items and using
  $\sum_j\mu_j=\E[\OPT]$, along with 
  $\sum_j y_j\le m$,
proves the theorem.
\end{proof}

\subsection{Proof of the signed one-level certificate} \label{sec:missingSignedOneLevel}

\begin{proof}[Proof of \Cref{lem:signed-one-level}]
The range and monotonicity of \(\coef\) are immediate.

First suppose \(e^\tau z\le1\).  Then
\[
\begin{aligned}
  f_\tau(z,q)
  &=(1-q)(1-e^\tau z)+zq  \\
  &=
  \frac12-\frac{(e^\tau-1)z}{2}
  +
  \big((e^\tau+1)z-1\big)\left(q-\frac12\right).
\end{aligned}
\]
Since \((e^\tau-1)z/2\le e^\tau-1=O(\tau)\), this gives the desired bound in this case.

Next suppose \(e^\tau z>1\) and \((e^\tau+1)z-1\le1\).  Then
\(f_\tau(z,q)=zq\) and \(z\le 2/(e^\tau+1)\).  Writing
\(e=q-1/2\), we have
\[
  zq-
  \left(\frac12+\big((e^\tau+1)z-1\big)e\right)
  =
  \frac{z-1}{2}+e(1-e^\tau z).
\]
Since \(1-e^\tau z<0\) and \(e\le1/2\), this is at least
\[
  \frac{z-1}{2}+\frac{1-e^\tau z}{2}
  =
  -\frac{(e^\tau-1)z}{2}
  \ge -O(\tau).
\]

Finally suppose \((e^\tau+1)z-1>1\).  Then \(\coef(z)=1\) and
\(z>2/(e^\tau+1)=1-O(\tau)\).  Also \(e^\tau z>1\), so
\(f_\tau(z,q)=zq\ge q-O(\tau)\).  This proves the signed bound in the final case.

For the last assertion, observe that
\[
  f_\tau(z,q)\ge \min\{1-q,e^{-\tau}q\}.
\]
Indeed, on \(0\le z\le e^{-\tau}\), the function is linear in \(z\) with endpoint values
\(1-q\) and \(e^{-\tau}q\), and for \(z\ge e^{-\tau}\) we have
\(f_\tau(z,q)=zq\ge e^{-\tau}q\).  If \(q\le1/2\), then
\(1-q\ge q\) and \(e^{-\tau}q\ge q-O(\tau)\).
\end{proof}

\section{Learning Cut-Calibrated Median Prices}
\label{sec:learning}

This section proves the learning guarantee used in the main theorem. For a fixed price grid, the statistical task is to use $T$ independent valuation-profile samples from the unknown product distribution $D$ to construct a distribution over base-price vectors with small expected cut-calibration error for every item. A \emph{training round} uses one such sample; it is not a real buyer's arrival. Algorithm~\ref{alg:main-training} gives the full implementation.

The mechanism runs one learner for every item, using the same sampled profile for all learners in a round. We fix an item $j$ and suppress $j$ from the notation whenever possible. In round $t$, before the sample and auxiliary price randomness are drawn, the learner chooses a distribution $\pi_t$ over grid levels. For this item, the training environment consists of that sampled profile and independently drawn prices for the other items, as specified in Algorithm~\ref{alg:main-training}. Holding this environment fixed and averaging over item $j$'s own perturbation gives the sale curve $y_t:[N]\to[0,1]$ defined in \Cref{sec:sale-curves}.

Let $\calF_t$ be the joint history of all learners before the round-$t$ sample and auxiliary prices are drawn, and define
\[
  Q_t(r):=\E[y_t(r)\mid\calF_t].
\]
The learner observes $y_t$ (or its oracle approximation), not $Q_t$; the latter is used only in the analysis. After $T$ rounds, deployment chooses one common round $U\sim\Unif([T])$ and draws each item's base level independently from its stored distribution for round $U$, with fresh perturbations and an independent valuation profile. Conditional on the completed training transcript, let $\Delta_j$ be the cut-calibration error of this final price distribution. Thus $\Delta_j$ is a random variable determined by the training transcript, and the expectation below is over that transcript.



\begin{theorem}[Learning cut calibration]\label{thm:learn-cut-calibration}
Assume that every observed sale curve satisfies
\[
  y_t(r)-y_t(r+1)\le \alpha
  \qquad\text{for all }t\text{ and }r<N.
\]
Using the path cut-regret learner of \Cref{thm:cut-regret},
\[
  \E[\Delta_j]
  \le
  \alpha+
  O\!\left(\sqrt{\frac{\log(2N)}{T}}\right).
\]
In particular, for the smoothed price grid in \Cref{sec:priceGrid}, we have
  $\alpha=\frac{\log(1+\beta)}{\tau}$.
Thus choosing \(\tau=\Theta(\eps)\), \(\beta=\Theta(\eps^2)\), and
\(T=\Theta(\eps^{-2}\log(2N))\) gives
  $\E[\Delta_j]\le O(\eps)$.
\end{theorem}

The main primitive is a cut-regret learner on the path of price levels.  After presenting
this primitive, we show that smoothed posted-price feedback produces exactly the losses
needed by the learner.

\subsection{Cut regret on the path}

For each cut \(b\in\{1,\ldots,N-1\}\), define the local map
\[
  \phi_b(r)=
  \begin{cases}
    r+1, & r\le b,\\
    r-1, & r>b.
  \end{cases}
\]
This map raises every level weakly below the cut and lowers every level above it.  These
are the only deviations we need because the welfare loss in \Cref{thm:cut-median} is
controlled by cut discrepancies.

\begin{theorem}[Path cut regret]\label{thm:cut-regret}
There is a full-information learner with the following guarantee.  In round
\(t\), the learner chooses \(\pi_t\) before \(\ell_t\) is revealed.  For every adaptive
sequence of scalar losses \(\ell_t:[N]\to\R\) with
 $ |\ell_t(r)-\ell_t(r')|\le 1$
   whenever $|r-r'|=1$,
\[
  \left(
  \max_{1\le b<N}
  \sum_{t=1}^T\sum_{r=1}^N
  \pi_t(r)\bigl(\ell_t(r)-\ell_t(\phi_b(r))\bigr)
  \right)_+
  \le
  O\!\left(\sqrt{T\log(2N)}\right).
\]
\end{theorem}

\begin{proof} The proof is inspired from swap regret to external regret reductions in \cite{BM-JMLR07,GGM-ICML08}.

Run deterministic multiplicative weights over the $N-1$ cut maps. 
In round $t$, let $\nu_t$ be its distribution over cuts.
This induces a Markov chain $K_t$
on \([N]\) by
\[
  K_t(r,a):=\Pr_{b\sim\nu_t}[\phi_b(r)=a].
\]
Let \(\pi_t\) be any stationary distribution of \(K_t\), which exists for every finite
Markov chain.

After observing \(\ell_t\), define the gain of cut \(b\) by
\[
  g_t(b):=\sum_r\pi_t(r)\bigl(\ell_t(r)-\ell_t(\phi_b(r))\bigr).
\]
Because \(\phi_b(r)\) is adjacent to \(r\) and adjacent loss differences are bounded by
one, \(g_t(b)\in[-1,1]\).  The external-regret guarantee gives
\[
  \max_b\sum_tg_t(b)
  \le
  \sum_t\sum_b\nu_t(b)g_t(b)
  +O\!\left(\sqrt{T\log(2N)}\right)
\]
for every realized sequence.

The on-policy term cancels by stationarity:
\[
\begin{aligned}
  \sum_b\nu_t(b)g_t(b)
  ~=~
  \sum_r\pi_t(r)\ell_t(r)
  -
  \sum_{r,a}\pi_t(r)K_t(r,a)\ell_t(a) 
  ~=~
  \sum_r\pi_t(r)\ell_t(r)-\sum_a\pi_t(a)\ell_t(a)
  ~=~ 0.
\end{aligned}
\]
Thus the positive part of the maximum cut gain is at most
$O(\sqrt{T\log(2N)})$ on every realized sequence, proving the theorem.
\end{proof}

\subsection{Sale curves and hop-distance losses} \label{sec:sale-curves}

We next explain why a single sampled valuation profile gives the full-information
one-dimensional feedback required by \Cref{thm:cut-regret}.

\begin{observation}[Price monotonicity]\label{lem:own-price}
Fix a valuation profile, a buyer order, all item prices except \(p_j\), and the fixed
global tie-breaking rule.  The final sold indicator of item \(j\) in the sequential
posted-price mechanism is nonincreasing as a function of \(p_j\).
\end{observation}

The proof of this observation is standard, so it's deferred to \Cref{app:own-price}.  Now fix a training round \(t\), the
sampled valuation profile, and all actual prices of items other than \(j\).  The only
remaining variable is item \(j\)'s own actual price.  By \Cref{lem:own-price}, there is a \emph{critical price} \(\Theta\in[0,\infty]\) such that \(j\) sells below \(\Theta\) and does not sell above \(\Theta\); the value at the threshold itself is irrelevant because the
perturbation is continuous.

With base level \(p_r\) and perturbation \(R=\exp(U)\), \(U\sim\Unif[0,\tau]\), define
\[
  y_t(r):=\Pr_R[p_rR\le\Theta].
\]
This is the perturbation-smoothed sale probability in training environment \(t\).  Since
\(p_{r+1}/p_r=1+\beta\),
\begin{equation}\label{eq:smooth}
  y_t(r)-y_t(r+1)
  ~\le~
  \alpha ~:=~ \frac{\log(1+\beta)}{\tau}.
\end{equation}
Indeed, the two sale events differ only when \(\log R\) lies in an interval of length
\(\log(1+\beta)\) inside a perturbation window of length \(\tau\).

It remains to convert this sale curve into scalar losses on the path.  For each
perturbation value \(R\), let  \emph{critical level} 
\[
  C_R:=\max\{r:p_rR\le\Theta\},
\]
with \(C_R=0\) if no level sells and \(C_R=N\) if all levels sell.  Then
  $y_t(r)=\Pr_R[C_R\ge r]$.

Define the \emph{hop-distance} loss
\[
  \ell_t(r):=\E_R[|r-C_R|].
\]
These losses have adjacent differences bounded by one.  Moreover, adjacent improvements
recover the sale probabilities:
\begin{equation}\label{eq:hop-identities}
  \ell_t(r)-\ell_t(r+1)=2y_t(r+1)-1
  \quad (r<N),
  \qquad
  \ell_t(r)-\ell_t(r-1)=1-2y_t(r)
  \quad (r>1).
\end{equation}

\subsection{From cut regret to fresh calibration}

We now show that cut regret on the hop-distance losses implies cut calibration of the
conditional, or fresh, sale curves.  Define the empirical cut error
\begin{equation}\label{eq:emp-delta}
  \widehat\Delta
  :=
  \left[
  \max_{1\le b<N}
  \left(
  \frac1T\sum_{t=1}^T\sum_{r\le b}\pi_t(r)\left(y_t(r)-\frac12\right)
  -
  \frac1T\sum_{t=1}^T\sum_{r>b}\pi_t(r)\left(y_t(r)-\frac12\right)
  \right)
  \right]_+ .
\end{equation}
Define the fresh cut statistic \(\Delta^Q\) by the same expression with \(Q_t\) replacing
\(y_t\).  Finally set
\[
  \CutReg_T:=
  \left(
  \max_{1\le b<N}\sum_{t=1}^T\sum_{r=1}^N
  \pi_t(r)\bigl(\ell_t(r)-\ell_t(\phi_b(r))\bigr)
  \right)_+ .
\]

\begin{lemma}[Cut regret to fresh calibration]
\label{lem:cut-calibration}\label{lem:cut-concentration}
For any adapted sequence of smoothed sale curves satisfying \eqref{eq:smooth},   
  $\widehat\Delta
  \le
  \alpha+\frac{\CutReg_T}{2T}$. 
Moreover,
  $\E[\Delta^Q]
  \le
  \E[\widehat\Delta]
  +
  O\!\left(\sqrt{\frac{\log(2N)}{T}}\right)$.
Consequently,
\[
  \E[\widehat\Delta]
  \le
  \alpha+\frac{\E[\CutReg_T]}{2T} \qquad \text{and}
  \qquad
  \E[\Delta^Q]
  \le
  \alpha+
  \frac{\E[\CutReg_T]}{2T}
  +
  O\!\left(\sqrt{\frac{\log(2N)}{T}}\right).
\]
\end{lemma}

\begin{proof}
Fix a cut \(b\).  Using \eqref{eq:hop-identities}, the cumulative gain of this cut is
\[
  G_b=
  \sum_{t=1}^T\left(
    \sum_{r\le b}\pi_t(r)(2y_t(r+1)-1)
    +
    \sum_{r>b}\pi_t(r)(1-2y_t(r))
  \right).
\]
By the smoothing inequality \eqref{eq:smooth}, \(y_t(r+1)\ge y_t(r)-\alpha\).  Hence
\[
  G_b
  \ge
  2\sum_{t=1}^T\left(
    \sum_{r\le b}\pi_t(r)\left(y_t(r)-\frac12\right)
    -
    \sum_{r>b}\pi_t(r)\left(y_t(r)-\frac12\right)
  \right)-2\alpha T.
\]
Let \(C_b\) be the empirical cut statistic inside \eqref{eq:emp-delta}.  The last
inequality gives
  $C_b\le \alpha+\frac{G_b}{2T}$.
If \(G_b\le0\), then \(C_b\le\alpha\).  If \(G_b>0\), then
\(G_b\le\CutReg_T\).  Therefore, for every cut \(b\), we have
  $C_b\le \alpha+\frac{\CutReg_T}{2T}$.
Taking the positive part of the maximum over cuts proves 
\[ \widehat\Delta
  \le
  \alpha+\frac{\CutReg_T}{2T} .
\]

We next compare empirical and fresh cut statistics.  For a cut \(b\), define
\(\sigma_b(r)=1\) if \(r\le b\) and \(\sigma_b(r)=-1\) if \(r>b\).  The difference
between the fresh and empirical cut statistic for \(b\) is
\[
  Z_b:=
  \frac1T\sum_{t=1}^T\sum_r
  \sigma_b(r)\pi_t(r)\big(Q_t(r)-y_t(r) \big).
\]
For each fixed \(b\), the summands form a martingale-difference sequence with absolute
value at most \(1/T\), because \(\pi_t\) is \(\calF_t\)-measurable,
\(\sum_r\pi_t(r)=1\), and \(Q_t(r),y_t(r)\in[0,1]\).  Azuma's inequality gives
\[
  \Prb[|Z_b|\ge u]\le2\exp(-\Omega(Tu^2)).
\]
A union bound over the \(N-1\) cuts and integration of the tail imply
  $\E\left[\max_b|Z_b|\right]
  =
  O\!\left(\sqrt{\frac{\log(2N)}{T}}\right)$.
The map \(x\mapsto[\max_bx_b]_+\) is \(1\)-Lipschitz in \(\ell_\infty\), so
  $\Delta^Q\le\widehat\Delta+\max_b|Z_b|$.
Taking expectations and
combining the two inequalities completes the proof.
\end{proof}

\subsection{Proof of the learning theorem}

\begin{proof}[Proof of \Cref{thm:learn-cut-calibration}]
Combining \Cref{thm:cut-regret} with \Cref{lem:cut-calibration} gives
\[
  \E[\Delta^Q]
  \le
  \alpha+
  O\!\left(\sqrt{\frac{\log(2N)}{T}}\right).
\]
It remains to identify \(\Delta^Q\) with the cut-calibration error \(\Delta_j\) of the
final price distribution.  The final mechanism chooses a round uniformly, so
  $\lambda_j(r)=\frac1T\sum_{t=1}^T\pi_t(r)$.
Moreover, whenever \(\lambda_j(r)>0\), conditioning on the event \(X_j=p_r\) biases the
chosen round proportionally to \(\pi_t(r)\).  Therefore
\[
  q_j(r)=
  \frac{\sum_{t=1}^T\pi_t(r)Q_t(r)}
       {\sum_{t=1}^T\pi_t(r)}.
\]
Hence
\[
  \lambda_j(r)\left(q_j(r)-\frac12\right)
  ~=~
  \frac1T\sum_{t=1}^T
  \pi_t(r)\left(Q_t(r)-\frac12\right).
\]
Thus the signed masses defining \(\Delta_j\) are exactly the signed masses defining
\(\Delta^Q\), and so \(\Delta_j=\Delta^Q\).  This proves the first claim.

For the smoothed price grid,
  $\alpha=\frac{\log(1+\beta)}{\tau}$.
With \(\tau=\Theta(\eps)\) and \(\beta=\Theta(\eps^2)\), we have
\(\alpha=O(\eps)\).  Taking \(T=C\eps^{-2}\log(2N)\) for a sufficiently large constant
\(C\) makes the regret and concentration term \(O(\eps)\), proving
\(\E[\Delta_j]\le O(\eps)\).
\end{proof}

\subsection{Proof of price monotonicity}
\label{app:own-price}

\begin{proof}[Proof of \Cref{lem:own-price}]
Compare two executions with \(p_j^-<p_j^+\).  We prove that sale at the high price
implies sale at the low price.  Consider the first buyer in the high-price execution who
chooses a bundle containing \(j\).  Up to this buyer, \(j\) is unsold in the high-price
execution.

Induct over earlier buyers.  Suppose item \(j\) has not yet sold in the low-price
execution.  Then the available sets of non-\(j\) items in the two executions are the same.
For the next buyer, every feasible bundle not containing \(j\) has the same utility in
both executions, while every feasible bundle containing \(j\) has larger utility in the
low-price execution.  Therefore either the low-price execution sells \(j\), or the fixed
tie-breaking rule chooses the same non-\(j\) bundle as in the high-price execution, and
the induction continues.

When we reach the first high-price buyer who buys \(j\), either \(j\) has already sold at
the low price, or this buyer faces the same availability of all non-\(j\) items.  Let
\(S^+\) be the bundle chosen in the high-price execution, so \(j\in S^+\).  For every
feasible bundle \(T\) not containing \(j\),
\[
  v_i(S^+)-p^+(S^+)
  \ge
  v_i(T)-p^+(T)
  =
  v_i(T)-p^-(T).
\]
Since \(p_j^-<p_j^+\) and \(j\in S^+\),
\[
  v_i(S^+)-p^-(S^+)
  >
  v_i(S^+)-p^+(S^+)
  \ge
  v_i(T)-p^-(T).
\]
Thus no utility-maximizing bundle at the low price can avoid \(j\).  The low-price
execution also sells item \(j\).
\end{proof}


\section{Completing the Main Theorem}\label{sec:main-proof}

We now prove \Cref{thm:intro-main}.  The proof combines the cut-calibrated welfare
bound from \Cref{thm:cut-median} with the learning guarantee from
\Cref{thm:learn-cut-calibration}.  We also need a short preprocessing step that finds a
value window.

For any valuation profile  \(v\) of XOS valuations, note that
\begin{equation}\label{eq:A-sandwich}
  \max_{i,j} v_i(\{j\}) ~\le~ \OPT(v) ~\le~ m \max_{i,j} v_i(\{j\}),
\end{equation}
where the lower bound follows by allocating the best singleton and the upper bound follows because at most \(m\) items are allocated and  supporting contribution of every allocated item is at most the largest singleton value.

\begin{lemma}[Price-window preprocessing]\label{lem:window}
Draw \(s=C\log(1/\delta)\) independent valuation-profile samples.  For each sampled
profile \(\hat v\), compute
$\max_{i,j}\hat v_i(\{j\})$,
and let \(b\) be the empirical \(2/3\)-quantile of these values, i.e.,
  $b:=\inf\{x:\widehat F(x)\ge 2/3\}$.
For a sufficiently large absolute constant \(C\), with probability at least \(1-\delta\),
the following hold.
\begin{enumerate}[label=(\roman*),leftmargin=2em]
  \item If \(b=0\), then posting zero prices, with empty-set tie-breaking among
  zero-utility bundles, obtains expected welfare at least \((7/12)\E[\OPT]\).

  \item If \(b>0\), then for every \(\rho\in(0,1)\), the window   $ B=2mb$ and  $L=\frac{\rho B}{m^2}$  
  satisfies
  \[
    \Prb_{v\sim D}[\OPT(v)\ge B]\le 5/12,
    \qquad
    mL\le 8\rho \E_{v \sim D}[\OPT(v)],
    \qquad
    B/L=m^2/\rho .
  \]
\end{enumerate}
\end{lemma}

The proof of \Cref{lem:window} is deferred to \Cref{app:window}.  We now use the lemma to
define the mechanism.

\subsection{Proof of \Cref{thm:intro-main}}

Run \Cref{lem:window} with failure probability \(\delta=\eps/10\).  If \(b=0\), use the
zero-price mechanism from the lemma.  Otherwise set
\[
  \rho=\Theta(\eps),
  \qquad
  B=2mb,
  \qquad
  L=\frac{\rho B}{m^2},
\]
with the hidden constant in \(\rho\) sufficiently small.  On the good preprocessing event,
\[
  \Prb[\OPT\ge B]\le5/12<1/2,
  \qquad
  mL=O(\eps\,\E[\OPT]),
  \qquad
  B/L=\poly(m,1/\eps).
\]

Set \(p_1=L\), and let
  $p_r=L(1+\beta)^{r-1}$
up to the first index \(N\) such that \(p_N\ge B\).  Thus \(p_N\le(1+\beta)B\).  Choose
  $\tau=\Theta(\eps)$ and $\beta=\Theta(\eps^2)$,
so that
\[
  \alpha=\frac{\log(1+\beta)}{\tau}=O(\eps)
  \qquad \text{and} \qquad
  N=O\!\left(\eps^{-2}\log(B/L)\right).
\]

\begin{algorithm}[t]
\caption{Training and deployment for \Cref{thm:intro-main}}
\label{alg:main-training}
\begin{enumerate}[leftmargin=2em]
  \item Initialize one copy of the cut-regret learner from
  \Cref{thm:cut-regret} for every item \(j\in M\).

  \item For each training round \(t=1,\ldots,T\):
  \begin{enumerate}[label=(\alph*),leftmargin=2em]
    \item Before seeing the round-\(t\) sample, every item learner \(j\)
    outputs a distribution \(\pi_{j,t}\) over the grid levels \([N]\).

    \item Draw one fresh valuation-profile sample
      $\widehat v^{\,t}\sim D$.
    This same sampled profile is used to generate feedback for all items.

    \item For every item \(j\), draw an auxiliary price environment for the
    other items:
    \[
      I^{(j,t)}_k\sim\pi_{k,t},
      \qquad
      V^{(j,t)}_k\sim\Unif[0,\tau],
      \qquad
      P^{(j,t)}_k=p_{I^{(j,t)}_k}e^{V^{(j,t)}_k},
      \qquad k\ne j.
    \]
    Conditional on \(\mathcal F_t\), take these auxiliary draws mutually
independent across all pairs \((j,k)\).

    \item Holding \(\widehat v^{\,t}\) and \(P^{(j,t)}_{-j}\) fixed, compute
    the smoothed sale curve
    \[
      y_{j,t}(r)
      :=
      \Pr_{V_j\sim\Unif[0,\tau]}
      \left[
        \begin{array}{c}
        \text{item \(j\) sells on profile \(\widehat v^{\,t}\)}\\
        \text{when \(P_j=p_re^{V_j}\) and \(P_{-j}=P^{(j,t)}_{-j}\)}
        \end{array}
      \right],
      \qquad r\in[N].
    \]
    Feed the corresponding the hop-distance loss defined in \Cref{sec:learning} to learner \(j\).
  \end{enumerate}

  \item After training, draw a common round
    $U\sim\Unif([T])$.
  Independently for every item \(j\), draw
  \[
    I_j\sim\pi_{j,U},
    \qquad
    V_j\sim\Unif[0,\tau],
    \qquad
    P_j=p_{I_j}e^{V_j}.
  \]
  Post the resulting price vector \(P\) to the real buyers in the fixed
  arrival order.
\end{enumerate}
\end{algorithm}

Algorithm~\ref{alg:main-training} uses one valuation-profile sample in each
training round.  All auxiliary price vectors are internal randomness and do
not count as additional samples.  The common round \(U\) may correlate the
items' base prices unconditionally; this is permitted by
\Cref{thm:cut-median}.

\begin{observation}[Consistency of training and deployment]
\label{obs:training-deployment}
Let \(\mathcal F_t\) be the joint history immediately before the valuation
sample and auxiliary prices of training round \(t\) are drawn, and define
  $Q_{j,t}(r):=\E[y_{j,t}(r)\mid\mathcal F_t]$.
Condition on the complete training transcript.  Whenever
\(\pi_{j,t}(r)>0\),
\[
  \Prb[
    j\text{ sells in the final mechanism}
    \mid U=t,\ I_j=r
  ]
  =
  Q_{j,t}(r).
\]
Consequently, conditional on the training transcript,
\[
  \lambda_j(r)=\frac1T\sum_{t=1}^T\pi_{j,t}(r)
\quad \text{and} \quad 
  \lambda_j(r)\left(q_j(r)-\frac12\right)
  =
  \frac1T\sum_{t=1}^T
  \pi_{j,t}(r)\left(Q_{j,t}(r)-\frac12\right).
\]
\end{observation}

\begin{proof}
Conditional on \(U=t\) and \(I_j=r\), the other items' base levels in the
final mechanism are independent draws from
\(\pi_{k,t}\), their perturbations are fresh draws from
\(\Unif[0,\tau]\), and the real valuation profile is a fresh draw from \(D\).
This is exactly the fresh environment whose conditional expectation defines
\(Q_{j,t}(r)\).  The two displayed identities then follow by averaging over
the uniformly random choice of \(U\) and applying Bayes' rule after
conditioning on \(I_j=r\).
\end{proof}

By \Cref{thm:learn-cut-calibration,obs:training-deployment} and taking
\(T=\Theta(\eps^{-2}\log(2N))\), we have
\begin{equation}\label{eq:main-delta-small}
  \E[\Delta_j]\le O(\eps)
  \qquad\text{for every item }j .
\end{equation}

Fix a preprocessing transcript $\mathcal T_0$ in the good event with $b>0$, so $L$ and $B$ are now fixed. The bound \eqref{eq:main-delta-small} holds conditional on every such transcript. Conditional on the completed training transcript, the final price distribution is independent of the fresh valuation profile, and the grid satisfies $p_1=L$, $p_N\ge B$, and $\Prb[\OPT\ge B]\le1/2$. Applying \Cref{thm:cut-median} conditionally on the completed training transcript and then averaging over the training gives
\[
\begin{aligned}
  \E[\ALG\mid\mathcal T_0]
  &\ge
  \Big(\frac12-O(\tau)\Big)\E[\OPT]
  -\sum_j\mu_j\E[\Delta_j\mid\mathcal T_0]
  -e^\tau Lm \\
  &\ge \Big(\frac12-O(\eps)\Big)\E[\OPT],
\end{aligned}
\]
where we use $\sum_j\mu_j=\E[\OPT]$, \eqref{eq:main-delta-small}, and $mL=O(\eps\E[\OPT])$. On a good preprocessing transcript with $b=0$, \Cref{lem:window} gives the stronger guarantee $(7/12)\E[\OPT]$. Thus the same target bound holds after averaging over all good preprocessing transcripts.

If the preprocessing event fails, welfare is nonnegative.  Since the failure probability is at most \(\eps/10\), after rescaling the internal accuracy parameter 
the unconditional expected welfare is still
\[
  \E[\ALG]\ge (1/2-\eps)\E[\OPT].
\]

It remains to count samples.  Since \(B/L=m^2/\rho\) and \(\rho=\Theta(\eps)\),
  $N=O\!\left(\eps^{-2}\log\frac{m+2}{\eps}\right)$.
Therefore
\[
  T=O\!\left(\eps^{-2}\log(2N)\right)
  =
  O\!\left(\eps^{-2}\log\frac{\log(m+2)}{\eps}\right).
\]
The preprocessing uses \(O(\log(10/\eps))=O(\log(2/\eps))\) additional samples, which is absorbed in the
same bound for \(\eps\in(0,1)\).

The mechanism is universally truthful: after the training samples are drawn, the final
randomized price vector is fixed independently of the real buyers' reports, and the online
phase is a deterministic posted-price mechanism conditioned on this randomness.

Finally, the mechanism is polynomial-time in the standard demand-oracle model.  The only
non-immediate step is computing the smoothed sale curve for an item on a sampled profile.
By \Cref{lem:critical-implementation}, this can be done to additive accuracy
\(\gamma\) using polynomially many demand queries whenever \(e^\tau p_N/L=\poly(m,1/\eps)\).
In our learned window,
\[
  \frac{e^\tau p_N}{L}
  ~\le~
  e^\tau(1+\beta)\frac BL
  ~=~
  \poly(m,1/\eps).
\]
Taking \(\gamma\) to be a sufficiently small constant multiple of \(\eps\) perturbs every
cut statistic and cut gain by at most \(O(\eps)\), which is absorbed into the final
\(O(\eps \E[\OPT])\) loss.  The remaining operations---updating the external-regret learner over
\(N-1\) cuts, forming cut prefix sums, and computing a stationary distribution of an
\(N\)-state Markov chain---are polynomial in \(n,m,1/\eps\).  This completes the proof of
\Cref{thm:intro-main}.

\subsection{Missing Proofs}\label{app:window}
\begin{proof}[Proof of \Cref{lem:window}]
Let 
  $A(v):=\max_{i,j}v_i(\{j\})$. For XOS valuations,
  $A(v)\le \OPT(v)\le mA(v)$.
Let \(F\) and \(\widehat F\) be the true and empirical CDFs of \(A\).  By the
Dvoretzky-Kiefer-Wolfowitz inequality, \(s=C\log(1/\delta)\) samples ensure
\[
  \sup_x |F(x)-\widehat F(x)|\le 1/12
\]
with probability at least \(1-\delta\).  Condition on this event.

By the quantile convention, \(\widehat F(b)\ge2/3\) and \(\widehat F(b^-)\le2/3\).  Hence
  $F(b)\ge7/12$ and 
  $F(b^-)\le3/4$.
Thus \(\Prb[A>b]\le5/12\).  
If \(b>0\), then
  $\Prb[A\ge b]=1-F(b^-)\ge1/4$,
and therefore
  \[\E[\OPT]\ge\E[A]\ge b/4.\]

First suppose \(b=0\).  Then \(\Prb[A>0]\le5/12\).  Since \(\OPT(v)\le mA(v)\), we have
  $\Prb[\OPT>0]\le5/12.$
At zero prices with empty-set tie-breaking among zero-utility bundles, if any item sells,
then some buyer has strictly positive value for some bundle, and hence the realized optimum
is positive.  Therefore, for every item \(j\),
\[
  \Prb[Y_j(0)=1]\le \Prb[\OPT>0]\le5/12,
  \qquad
  \Prb[Y_j(0)=0]\ge7/12.
\]
Applying the FGL certificate at zero prices gives welfare at least
\[
  \sum_j\mu_j\Prb[Y_j(0)=0]
  ~\ge~
  (7/12)\sum_j\mu_j
  ~=~
  (7/12)\E[\OPT].
\]
Now suppose \(b>0\).  If \(\OPT\ge B=2mb\), then
  $A\ge \OPT/m\ge 2b>b$.
Hence
  $\Prb[\OPT\ge B]\le\Prb[A>b]\le5/12$.
Finally,
\[
  mL=\frac{\rho B}{m}=2\rho b\le8\rho\,\E[\OPT],
\]
and \(B/L=m^2/\rho\) is immediate.
\end{proof}


\begin{lemma}[Computing a smoothed sale curve]\label{lem:critical-implementation}
Fix a sampled valuation profile, an item \(j\), and actual prices of all items other than
\(j\).  Let \(P_{\max}=e^\tau p_N\).  For every \(\gamma\in(0,1)\), one can compute
numbers \(\widetilde y(r)\) satisfying
\[
  \max_{r\in[N]}|\widetilde y(r)-y(r)|\le\gamma,
  \qquad
  y(r)=\Pr_{U\sim\Unif[0,\tau]}[\text{\(j\) sells at actual price }p_re^U],
\]
using
 $ O\!\left(
    n\log\!\left(2+\frac{\log(P_{\max}/L)}{\gamma\tau}\right)
  \right)$
tie-broken demand queries and \(O(N)\) arithmetic operations.
\end{lemma}

\begin{proof}
A query at actual price \(x\) is implemented by simulating the posted-price mechanism on
the fixed sampled profile with item \(j\) priced at \(x\) and all other prices fixed.
This uses one demand query per buyer.  By \Cref{lem:own-price}, the resulting sold
indicator \(S_j(x)\) is nonincreasing in \(x\).

If \(S_j(L)=0\), then no possible actual price \(p_re^U\ge L\) sells item \(j\), so the
curve is identically zero.  If \(S_j(P_{\max})=1\), then every possible actual price
\(p_re^U\le P_{\max}\) sells item \(j\), so the curve is identically one.  Otherwise
\(S_j(L)=1\) and \(S_j(P_{\max})=0\).  Let
\[
  \Theta:=\sup\{x:S_j(x)=1\}\in[L,P_{\max}].
\]
Binary search on \(\log x\) maintains a bracket \([\ell,h]\) with \(S_j(\ell)=1\),
\(S_j(h)=0\), and \(\Theta\in[\ell,h]\).  After
\[
  O\!\left(
    \log\!\left(2+\frac{\log(P_{\max}/L)}{\gamma\tau}\right)
  \right)
\]
iterations, the bracket has logarithmic length at most \(\gamma\tau\).  Pick any
\(\widetilde\Theta\in[\ell,h]\) and output
\[
  \widetilde y(r)=
  \left[
    \frac{\log\widetilde\Theta-\log p_r}{\tau}
  \right]_{[0,1]},
  \qquad
  [x]_{[0,1]}:=\min\{1,\max\{0,x\}\}.
\]
The exact curve has the same formula with \(\Theta\) in place of
\(\widetilde\Theta\), except at a single threshold point of the continuous perturbation
distribution.  Since the clipping map is \(1\)-Lipschitz and
\(|\log\widetilde\Theta-\log\Theta|\le\gamma\tau\), the error is at most \(\gamma\) for
every grid level.
\end{proof}

\section{Single-Sample and Random-Order Consequences}
\label{sec:SSRO}\label{sec:one-sample}

We now prove \Cref{cor:intro-one-sample,cor:intro-random-order}.  The same
cut-calibrated learner used in the main theorem also gives direct one-sample and
random-order consequences.  In both settings the mechanism mixes two branches: a sparse
item-pricing branch, which captures contributions at a chosen scale, and a standard
grand-bundle branch, which captures the upper tail.

\paragraph{Scale estimation and the buyer sets.}
In both settings, $\mathrm{STAT}$ is the set of buyers used only to estimate the scale, and $\mathrm{FIXED}:=[n]\setminus\mathrm{STAT}$ is the remaining buyer pool. In the one-sample setting, write $s=(s_i)_i$ for the sample profile and $v=(v_i)_i$ for the real profile. Put each buyer independently into $\mathrm{STAT}$ with probability $1/2$ and set $\Gamma:=\max_{i\in\mathrm{STAT}}s_i(M)$. In the random-order setting with fixed profile $u$, use the first $\lfloor n/2\rfloor$ arrivals as $\mathrm{STAT}$ and set $\Gamma:=\max_{i\in\mathrm{STAT}}u_i(M)$; here $\mathrm{FIXED}$ is the post-STAT suffix. In either case, an empty maximum is zero, STAT buyers receive nothing and pay nothing, and the sparse branch returns no allocation if $\Gamma=0$.

The sparse branch uses a constant internal accuracy parameter.  Fix absolute constants in
the following order: \(C_B\) large, \(\rho,\tau,\beta>0\) sufficiently small, and then
\(C,C_K\) sufficiently large.  Given a scale \(\Gamma>0\), the sparse branch uses the
window
\[
  L=\frac{\rho\Gamma}{m},
  \qquad
  B=2e^\tau C_B\Gamma,
\]
and a geometric grid from \(L\) to at least \(B\) with ratio \(1+\beta\).  Since
\(B/L=O(m)\), the number of grid levels is
  $N=O(\log(m+2))$.
Set
  $H=C\log(2N)$ and $
  K=\lceil C_KH\rceil$, which means 
  $K=O(\log\!\log (m+2))$.

\paragraph{Zero-padded roles (one sample).}
Each buyer in $\mathrm{FIXED}$ independently receives a uniformly random label $\ell_i\in\{0,1,\ldots,K\}$. Label $0$ is the allocation role, labels $1,\ldots,H$ are training roles, and the remaining labels are ignored. Training role $t$ uses the sample valuations $s_i$ of its buyers; allocation role $0$ uses the real valuations $v_i$ of its buyers. All executions preserve the original fixed buyer order.

For the analysis, define, for each $i\in\mathrm{FIXED}$,
\[
  U_i:=
  \begin{cases}
    v_i, & \ell_i=0,\\
    s_i, & \ell_i\ne0.
  \end{cases}
\]
The \emph{zero-padded role profile} $V^a$ has valuation $U_i$ at each original buyer position with label $a$ and the zero valuation at every other position. In particular, $V^t$ is the sampled training profile for role $t$, while $V^0$ is the real allocation profile. Conditional on the STAT/FIXED split, the coupling argument below shows that conditioning additionally on $U$ and $\Gamma$ does not bias the labels.

\paragraph{Chronological blocks (random order).}
Add the minimum number of zero-valued dummy buyers needed to make the post-STAT suffix size divisible by $K+1$, inserting the dummy positions uniformly at random among the suffix positions, independently of the real buyers' order. Partition the padded suffix into $K+1$ equal-size chronological blocks. Blocks $1,\ldots,H$ are training blocks, block $K+1$ is the allocation block, and all other blocks are ignored. If the common block size is $s_{\rm blk}$, then for each item and each training block, independently draw $R\sim\Unif\{0,\ldots,s_{\rm blk}-1\}$ and use the first $R$ positions of that block for the item's feedback. These prefix lengths are independent of the block order and all valuations.

\paragraph{Training, storage, and allocation.}
Initialize one cut-regret learner from \Cref{thm:cut-regret} per item. At the start of round $t$, using only the scale and previously exposed training roles or blocks, learner $j$ outputs $\pi_{j,t}$. Before the current role or block is exposed, independently draw $I_{j,t}\sim\pi_{j,t}$ and $Z_{j,t}\sim\Unif[0,\tau]$ for every item, and store the actual price vector
\[
  P_{j,t}:=p_{I_{j,t}}e^{Z_{j,t}}.
\]
For feedback to learner $j$, hold $P_{-j,t}$ fixed and compute the perturbation-smoothed sale curve while varying only $j$'s base level and perturbation. Use the entire zero-padded role in the one-sample implementation and the random block prefix in the random-order implementation. Each simulation starts with all items available. Feed the corresponding hop-distance loss defined in \Cref{sec:sale-curves} to learner $j$. In the one-sample implementation, labels and sampled valuations from unexposed roles are not revealed to the learners.

After training, choose $t^\star\sim\Unif([H])$ and post the stored vector $P_{t^\star}$ to the allocation-role or allocation-block buyers, starting with all items available. Here the stored actual prices are reused, rather than drawing new prices from the stored learner distributions. Buyers outside the allocation role or block receive nothing and pay nothing.

We use two black-box lemmas, both proved in \Cref{app:ssro}.  The first says that, once a
scale \(\Gamma\) is fixed and all benchmark buyers have grand-bundle value at most
\(C_B\Gamma\), the sparse branch captures a \(1/(K+1)\) fraction of the benchmark up to
the negligible grid-floor loss.

The profile $u'$ in the first lemma is used only as a benchmark: the sparse mechanism runs on the full served profile $u$, without deleting buyers according to this benchmark truncation.

\begin{lemma}[Sparse scale lemma]\label{lem:sparse-scale}
Let \(u\) be the valuation profile seen by the sparse branch, and let \(u'\) be obtained
from \(u\) by replacing some buyers by the zero valuation.  Suppose every nonzero buyer in
\(u'\) satisfies \(u'_i(M)\le C_B\Gamma\).  In the zero-padded role implementation, where
the expectation is over the random labels and the mechanism's internal randomness,
\[
  \E[\ALG_{\rm sparse}]
  \ge
  \frac{c}{K+1}\bigl(\OPT(u')-C'\rho\Gamma\bigr).
\]
The same guarantee holds for the chronological-block implementation in the random-order
setting, with expectation over the random block order and the mechanism's internal
randomness, conditioned on the scale \(\Gamma\) and on the buyers remaining after the STAT
block.
\end{lemma}

The second lemma says that the assumption in \Cref{lem:sparse-scale} loses only the part
already handled by the grand-bundle branch.

\begin{lemma}[Grand-bundle scale reduction]\label{lem:gb-reduction}
For \(C_B\) sufficiently large, the following hold.

\begin{itemize}
    \item 
In the one-sample setting, let
  $G:=\E_{v\sim D}\left[\max_i v_i(M)\right]$.
There is a universally truthful grand-bundle posted-price branch, using the same one
sample, with expected welfare \(\Omega(G)\).  Moreover, for the STAT scale \(\Gamma\) and
the coupled fixed-buyer profile used below, if \(U^{\le}\) is obtained by zeroing out all
buyers \(i\) with \(U_i(M)>C_B\Gamma\), then
\[
  \E[\OPT(U^{\le})]
  \ge
  \Omega(\E[\OPT])-O(G).
\]

\item In the random-order setting, fix a valuation profile \(u\).
There is a universally truthful grand-bundle secretary branch with expected welfare
\(\Omega(\max_i u_i(M))\).  Moreover, for the STAT scale \(\Gamma\), if \(u^{\le}\) is obtained from
the post-STAT buyers by zeroing out all buyers \(i\) with \(u_i(M)>C_B\Gamma\), then
\[
  \E[\OPT(u^{\le})]
  \ge
  \Omega(\OPT(u))-O(\max_i u_i(M)),
\]
where the expectation is over the random order.
\end{itemize}
\end{lemma}

\subsection{One sample}
We prove \Cref{cor:intro-one-sample}. The mechanism flips a fair coin between the grand-bundle branch from \Cref{lem:gb-reduction} and the zero-padded sparse branch specified above. We analyze the sparse branch using the coupled profile $U$ defined above.

We will condition on the realized profile \(U\), but not on the labels.  The following
simple Bayes calculation shows that this conditioning does not bias the labels.  For every
measurable event \(A\) and every label \(a\in\{0,\ldots,K\}\), conditional on
\(i\in\mathrm{FIXED}\) and on \(\Gamma\),
\[
  \Pr[\ell_i=a,\ U_i\in A]
  =
  \frac1{K+1}D_i(A).
\]
Thus \(U_i\sim D_i\) and \(U_i\) is independent of \(\ell_i\).  Since the pairs
\((s_i,v_i,\ell_i)\) are independent across fixed buyers, and \(\Gamma\) depends only on
STAT samples, the labels of buyers in \(\mathrm{FIXED}\) remain independent and uniform
even after conditioning on \(\Gamma\) and the realized profile \(U\).

Let \(U^{\le}\) be obtained from \(U\) by replacing every buyer with
\(U_i(M)>C_B\Gamma\) by the zero valuation.  Conditional on \(\Gamma\) and \(U\), the
sparse branch is exactly the zero-padded role implementation of
\Cref{lem:sparse-scale}, run on the fixed profile \(U\) and benchmarked against
\(U^{\le}\).  Therefore
\[
  \E[\ALG_{\rm sparse}\mid \Gamma,U]
  \ge
  \frac{c}{K+1}\bigl(\OPT(U^{\le})-C'\rho\Gamma\bigr).
\]
Taking expectations, using \(\Gamma\le \OPT(s)\), and applying
\Cref{lem:gb-reduction}, we get
\[
  \E[\ALG_{\rm sparse}]
  \ge
  \Omega\!\left(\frac{\E[\OPT]}{K+1}\right)
  -
  O\!\left(\frac{G}{K+1}\right)
  -
  O\!\left(\frac{\rho\,\E[\OPT]}{K+1}\right).
\]
Choose \(\rho\) sufficiently small.  The negative \(G\)-term is covered by mixing in the
grand-bundle branch, whose expected welfare is \(\Omega(G)\).  Therefore the mixed
mechanism satisfies
\[
  \E[\ALG]
  \ge
  \Omega\!\left(\frac{\E[\OPT]}{K+1}\right)
  =
  \Omega\!\left(\frac{\E[\OPT]}{\log\!\log (m+2)}\right).
\]

The mechanism is universally truthful.  In the sparse branch, STAT, training, and ignored
buyers receive nothing and pay nothing; allocation-role buyers face prices fixed before
they are served.  The grand-bundle branch is also a posted-price mechanism.  This proves
\Cref{cor:intro-one-sample}.

\subsection{Random order}

We prove \Cref{cor:intro-random-order}.  Fix an XOS valuation profile \(u\), and expose a
uniformly random arrival order.  The mechanism again mixes a grand-bundle branch with a
sparse item-pricing branch.

The grand-bundle branch is the classical secretary threshold rule applied to the grand-bundle values $u_i(M)$, using independent random priorities to break ties between equal values. By \Cref{lem:gb-reduction}, it obtains expected welfare $\Omega(x^*)$, where $x^*=\max_i u_i(M)$. The sparse branch uses the chronological-block implementation specified above.

Let \(u^{\le}\) be the post-STAT profile after zeroing out every buyer \(i\) with
\(u_i(M)>C_B\Gamma\).  Conditional on \(\mathrm{STAT}\), \(\Gamma\), and the remaining
buyers, the sparse branch is the chronological-block implementation of
\Cref{lem:sparse-scale}, run on the full post-STAT profile and benchmarked against
\(u^{\le}\).  Therefore
\[
  \E[\ALG_{\rm sparse}\mid \mathrm{STAT},\Gamma]
  \ge
  \frac{c}{K+1}\bigl(\OPT(u^{\le})-C'\rho\Gamma\bigr).
\]
Since \(\Gamma\le x^*\le \OPT(u)\), taking expectation over the random order and applying
\Cref{lem:gb-reduction} gives
\[
  \E[\ALG_{\rm sparse}]
  \ge
  \Omega\!\left(\frac{\OPT(u)}{K+1}\right)
  -
  O\!\left(\frac{x^*}{K+1}\right)
  -
  O\!\left(\frac{\rho\,\OPT(u)}{K+1}\right).
\]
Choose \(\rho\) sufficiently small.  The negative \(x^*\)-term is covered by the
grand-bundle secretary branch.  Hence the mixed mechanism satisfies
\[
  \E[\ALG]
  \ge
  \Omega\!\left(\frac{\OPT(u)}{K+1}\right)
  =
  \Omega\!\left(\frac{\OPT(u)}{\log\!\log (m+2)}\right).
\]

All prices offered to allocation-block buyers are fixed before those buyers are served,
and all STAT, training, ignored, and dummy buyers receive nothing and pay nothing.  Thus
the mechanism is universally truthful, proving \Cref{cor:intro-random-order}.

\medskip
\medskip

\paragraph{Acknowledgments.}
The authors used GPT-5.5 Pro for assistance with writing and to  explore proof strategies. The authors have checked the arguments and take full responsibility for all content of the paper.

\appendix


\section{Deferred Proofs for \Cref{sec:SSRO}}
\label{app:ssro}

\subsection{Grand-bundle scale reduction}
\label{app:gb-reduction}

We first record the elementary tail estimates used by the scale reduction.

\begin{lemma}[Independent scale tail]\label{lem:ind-scale-tail}
Let \(X_1,\ldots,X_n\) be independent nonnegative random variables, and let
\(S_1,\ldots,S_n\) be independent copies.  Independently put each index in
\(\mathrm{STAT}\) with a fixed constant probability \(q\in(0,1)\), and set
  $\Gamma:=\max_{i\in\mathrm{STAT}}S_i$,
with \(\Gamma=0\) if \(\mathrm{STAT}=\emptyset\).  For every sufficiently large constant
\(C_0=C_0(q)\),
\[
  \E\left[
    \sum_{i\notin\mathrm{STAT}}X_i\1\{X_i>C_0\Gamma\}
  \right]
  \le
  O_q\!\left(\E[\max_iX_i]\right).
\]
\end{lemma}

\begin{proof}
Let
\[
  A(x):=\sum_i\Prb[X_i>x],
  \qquad
  X^*:=\max_iX_i.
\]
For every deterministic \(a\ge0\) and every nonnegative random variable \(X\),
\[
  \E[X\1\{X>a\}]
  \le
  2\int_{a/2}^{\infty}\Prb[X>x]\,dx .
\]
Therefore, for \(C_0\ge4\),
\[
\begin{aligned}
  \E\left[
    \sum_{i\notin\mathrm{STAT}}X_i\1\{X_i>C_0\Gamma\}
  \right]
  &\le
  2\int_0^\infty A(x)\Prb[\Gamma<2x/C_0]\,dx .
\end{aligned}
\]
Since \(2/C_0\le1\), the tail sum of the copies \(S_i\) at level \(2x/C_0\) is at least
\(A(x)\).  Hence
$  \Prb[\Gamma<2x/C_0]\le \exp(-qA(x))$.
Also,
\[
  \Prb[X^*>x]
  =
  1-\prod_i(1-\Prb[X_i>x])
  \ge
  1-e^{-A(x)}.
\]
The scalar inequality
  $A e^{-qA}\le C(q)(1-e^{-A})$ for 
 $A\ge0$
implies
\[
  A(x)\Prb[\Gamma<2x/C_0]
  \le
  C(q)\Prb[X^*>x].
\]
Integrating over \(x\) proves the lemma.
\end{proof}

\begin{lemma}[Random-subset scale tail]\label{lem:subset-scale-tail}
Fix nonnegative numbers \(x_1,\ldots,x_n\), and put each index in
\(\mathrm{STAT}\) independently with a fixed constant probability \(q\in(0,1)\).  Let
\[
  \Gamma:=\max_{i\in\mathrm{STAT}}x_i,
  \qquad
  x^*:=\max_i x_i .
\]
For every constant \(C_0>1\),
\[
  \E_{\mathrm{STAT}}\left[
    \sum_{i\notin\mathrm{STAT}}x_i\1\{x_i>C_0\Gamma\}
  \right]
  \le
  O_q(x^*).
\]
The same conclusion holds, up to changing the constant, when \(\mathrm{STAT}\) is a
uniformly random fixed-size subset of constant density.
\end{lemma}

\begin{proof}
Order the indices so that \(x_1\ge x_2\ge\cdots\ge x_n\).  If index \(i\) contributes,
then \(i\notin\mathrm{STAT}\) and every index \(h\le i\) is also outside
\(\mathrm{STAT}\), because \(x_h\ge x_i>C_0\Gamma \geq \Gamma\).  Thus the contribution
probability is at most \((1-q)^i\), and
\[
  \E\left[
    \sum_{i\notin\mathrm{STAT}}x_i\1\{x_i>C_0\Gamma\}
  \right]
  \le
  \sum_{i\ge1}x_i(1-q)^i
  \le
  x^*\sum_{i\ge1}(1-q)^i
  =
  O_q(x^*).
\]
For a fixed-size subset of constant density, the probability that the first \(i\) ordered
indices are all outside \(\mathrm{STAT}\) is at most \(\exp(-\Omega(i))\), giving the same
bound.
\end{proof}

\begin{proof}[Proof of \Cref{lem:gb-reduction}]
For the one-sample grand-bundle branch, apply the single-sample single-choice prophet inequality of Rubinstein, Wang, and Weinberg~\cite{RubinsteinWW20} to the independent one-dimensional random variables \(v_i(M)\).  This gives a posted price for the grand bundle, computed from the sampled grand-bundle values, with expected welfare
\[
  \Omega\!\left(\E_{v\sim D}[\max_i v_i(M)]\right)=\Omega(G).
\]

We now prove the scale-loss bound in the one-sample setting.  We use the following
distributionally equivalent coupling of the fixed-buyer profile.  First draw an
independent full profile \(\bar U=(\bar U_i)_i\sim D\), independently of the
\(\mathrm{STAT}/\mathrm{FIXED}\) split, and then set \(U_i=\bar U_i\) for
\(i\in\mathrm{FIXED}\).  This has the same joint distribution as the coupled profile used
in the one-sample mechanism, because for every fixed buyer the role label is independent
of \(s_i,v_i\), and both \(s_i\) and \(v_i\) are independent draws from \(D_i\).

Let \((O_i)_i\) be an optimal allocation for \(\bar U\), with supporting XOS clauses.
Keeping only the bundles \(O_i\) of buyers \(i\in\mathrm{FIXED}\) is feasible for the
restricted profile \(U\).  Since each buyer is in \(\mathrm{FIXED}\) with constant
probability independently of \(\bar U\),
\[
  \E[\OPT(U)]
  ~\ge~
  \E\!\left[\sum_{i\in\mathrm{FIXED}} \bar U_i(O_i)\right]
  ~=~
  \Pr[i\in\mathrm{FIXED}]\,\E[\OPT]
  ~=~
  \Omega(\E[\OPT]).
\]
Let \(U^{\le}\) be obtained by zeroing out buyers with \(U_i(M)>C_B\Gamma\).  The welfare
lost by this deletion is at most
\[
  \sum_{i\in\mathrm{FIXED}}U_i(M)\1\{U_i(M)>C_B\Gamma\},
\]
because the total supporting contribution assigned to buyer \(i\) in any XOS allocation
is at most \(U_i(M)\).  Applying \Cref{lem:ind-scale-tail}, with \(C_B\) sufficiently
large, bounds the expectation of this loss by \(O(G)\).  Hence
\[
  \E[\OPT(U^{\le})]\ge \Omega(\E[\OPT])-O(G).
\]

For random order, the grand-bundle branch is the classical secretary threshold rule
applied to the grand-bundle values \(u_i(M)\).  With constant probability the maximum
grand-bundle buyer appears after the sample phase and the best earlier grand-bundle value
lies in the sample phase; on this event the mechanism selects the maximum grand-bundle
buyer.  Thus the branch obtains \(\Omega(x^*)\).

It remains to prove the random-order scale-loss bound.  Removing the initial STAT block
preserves a constant fraction of the fixed optimum in expectation: keep from a fixed
optimal allocation only the buyers outside the STAT block.  Now delete from the remaining
buyers all \(i\) with \(u_i(M)>C_B\Gamma\).  The lost welfare is at most
\[
  \sum_{i\notin\mathrm{STAT}}u_i(M)\1\{u_i(M)>C_B\Gamma\}.
\]
By \Cref{lem:subset-scale-tail}, its expectation is \(O(x^*)\).  Therefore
\[
  \E[\OPT(u^{\le})]\ge \Omega(\OPT(u))-O(x^*).
\]
This proves the lemma.
\end{proof}

\subsection{Sparse scale lemma}
\label{app:sparse-scale}

We prove \Cref{lem:sparse-scale}.  Fix a profile \(u\) seen by the sparse branch and a
benchmark profile \(u'\) obtained from \(u\) by zeroing out some buyers.  Fix an optimal
allocation for \(u'\) and supporting XOS clauses.  If item \(j\) is allocated to buyer
\(i\), write \(i(j)=i\) and let \(w_j\) be its supporting contribution; if \(j\) is
unallocated, set \(w_j=0\).  Thus
  $\sum_j w_j=\OPT(u')$.
Because every nonzero buyer in \(u'\) has grand-bundle value at most \(C_B\Gamma\), every
allocated contribution satisfies \(w_j\le C_B\Gamma=B/(2e^\tau)\).

Define the midrange benchmark
\[
  \mathrm{Mid}_\Gamma(u')
  :=
  \sum_{j:\,2e^\tau L\le w_j\le B/(2e^\tau)}w_j .
\]
The omitted small contributions have total value at most
  $2e^\tau Lm=O(\rho\Gamma)$.
There are no large contributions by the scale bound.  Hence
\begin{equation}\label{eq:mid-large-enough}
  \mathrm{Mid}_\Gamma(u')\ge \OPT(u')-O(\rho\Gamma).
\end{equation}
It remains to show that the sparse branch captures an \(\Omega(1/(K+1))\) fraction of
\(\mathrm{Mid}_\Gamma(u')\).

We use the stored-round convention specified at the start of \Cref{sec:SSRO}. In training round $t$, the item learners output their distributions and the mechanism draws and stores $P_t$ before the current role or block is exposed. For feedback to item $j$, the coordinates $P_{-j,t}$ remain fixed and only $j$'s own base level and perturbation are varied. By \Cref{lem:own-price}, this gives a valid perturbation-smoothed sale curve.

\subsubsection*{Zero-padded roles}

We first prove the zero-padded role guarantee.

\begin{lemma}[Role exposure]\label{lem:role-exposure}
Fix a training role \(t\le H\).  Let the labels \(\ell_i\) be independent and uniform on
\(\{0,1,\ldots,K\}\), and let \(\mathcal H_t\) be the history before role \(t\) is
exposed: the learner's private randomness, the stored price vector \(P_t\), and the labels
and valuations in roles \(1,\ldots,t-1\).  Then:
\begin{enumerate}[label=(\alph*),leftmargin=2em]
  \item If \(\Phi\) is any bounded, or nonnegative integrable, measurable function of a
  price vector and a role profile, and \(P_t\) is \(\mathcal H_t\)-measurable, then
  \[
    \E[\Phi(P_t,V^0)]=\E[\Phi(P_t,V^t)].
  \]

  \item Fix a buyer \(i\).  If \(X_t\in[0,1]\) is \(\mathcal H_t\)-measurable, then
  \[
    \E[\1\{\ell_i=t\}X_t]
    \ge
    \frac1{K+1}\left(\E[X_t]-\frac{t-1}{K+1}\right).
  \]

  \item If \(Z_t\in[0,1]\) is measurable with respect to \(\mathcal H_t\), the labels and
  valuations in role \(t\) of buyers other than \(i\), and any additional randomness
  independent of \(\ell_i\) conditional on \(\mathcal H_t\), then
  \[
    \E[\1\{\ell_i=t\}Z_t]
    \le
    \frac1{K-t+2}\E[Z_t]
    \le
    \frac{1+O(H/K)}{K+1}\E[Z_t],
  \]
  where the last inequality uses \(t\le H\) and \(K\ge2H\).
\end{enumerate}
\end{lemma}

\begin{proof}
Condition on \(\mathcal H_t\).  For every buyer not already revealed in roles
\(1,\ldots,t-1\), the conditional distribution of her label is uniform over the remaining
labels.  In particular, labels \(0\) and \(t\) have equal conditional probabilities.
Swapping the names of labels \(0\) and \(t\) preserves the conditional law and leaves
\(P_t\) fixed, proving part (a) for bounded \(\Phi\); the nonnegative case follows by
truncation.

For part (b), let \(A_{i,t}\) be the event that buyer \(i\) has not appeared in roles
\(1,\ldots,t-1\).  Conditional on \(\mathcal H_t\), the probability of \(\ell_i=t\) is
zero on \(A_{i,t}^c\) and is \(1/(K-t+2)\) on \(A_{i,t}\).  Since
\(1/(K-t+2)\ge1/(K+1)\),
\[
  \E[\1\{\ell_i=t\}X_t]
  \ge
  \frac1{K+1}\E[X_t\1_{A_{i,t}}]
  \ge
  \frac1{K+1}\left(\E[X_t]-\Prb[A_{i,t}^c]\right).
\]
Here \(\Prb[A_{i,t}^c]=(t-1)/(K+1)\).

For part (c), conditional on \(\mathcal H_t\) and \(A_{i,t}\), the event
\(\{\ell_i=t\}\) is independent of \(Z_t\) by the measurability assumption.  Therefore
\[
  \E[\1\{\ell_i=t\}Z_t]
  =
  \frac1{K-t+2}\E[Z_t\1_{A_{i,t}}]
  \le
  \frac1{K-t+2}\E[Z_t].
\]
The comparison with \(1/(K+1)\) follows from \(t\le H\) and \(K\ge2H\).
\end{proof}

\begin{lemma}[Midrange capture for zero-padded roles]
\label{lem:midrange-roles}
In the zero-padded role implementation,
\[
  \E[\ALG_{\rm sparse}]
  \ge
  \frac{c}{K+1}\,\mathrm{Mid}_\Gamma(u')
\]
for an absolute constant \(c>0\).
\end{lemma}

\begin{proof}
Let \(V^t\) denote the role-\(t\) profile induced by the served profile \(u\).  Since the
final sparse branch chooses \(t\sim\Unif([H])\) and posts \(P_t\) to role \(0\),
\[
  \E[\ALG_{\rm sparse}]
  ~=~
  \frac1H\sum_{t=1}^H\E[\Welfare(P_t,V^0)]
  ~=~
  \frac1H\sum_{t=1}^H\E[\Welfare(P_t,V^t)],
\]
where the second equality is \Cref{lem:role-exposure}(a).

Fix a midrange item \(j\), write \(i=i(j)\), \(w=w_j\), and let \(b=b(j)\) be the largest
grid index such that \(p_be^\tau\le w/2\).  The midrange definition implies
\(1\le b<N\).  If \(r\le b\), every perturbed actual price of item \(j\) is at most
\(w/2\); if \(r>b\), the base price is at least \(w/(2e^\tau)\).

For round \(t\), hold \(P_{-j,t}\) fixed.  Let \(Y_{j,t}(r)\) be the probability, over
item \(j\)'s own perturbation, that item \(j\) sells by the end of role \(t\) when its
base level is \(r\).  Let \(A^i_{j,t}(r)\) be the analogous probability that \(j\) sells
before buyer \(i\) is reached.  Define
\[
  \Lambda_{j,t}:=\sum_{r\le b}\pi_{j,t}(r), \qquad 
  S^{\rm low}_{j,t}:=\sum_{r\le b}\pi_{j,t}(r)Y_{j,t}(r),
  \qquad
  S^{\rm high}_{j,t}:=\sum_{r>b}\pi_{j,t}(r)Y_{j,t}(r),
\]
and
\[
  A^{\rm low}_{j,t}:=\sum_{r\le b}\pi_{j,t}(r)A^i_{j,t}(r).
\]
Then \(A^{\rm low}_{j,t}\le S^{\rm low}_{j,t}\).

We claim that for an absolute constant \(\kappa>0\),
\begin{equation}\label{eq:role-item-cert}
  \E[\Welfare(P_t,V^t)]
  \ge
  \kappa
  \sum_{j\in J_{\rm mid}}w_j\,
  \E\!\left[
    S^{\rm high}_{j,t}
    +
    \1\{\ell_{i(j)}=t\}(\Lambda_{j,t}-A^{\rm low}_{j,t})
  \right],
\end{equation}
where \(J_{\rm mid}\) is the set of midrange benchmark items.  Indeed, high-level sales
of \(j\) give revenue at least \(w_j/(2e^\tau)\).  For utility, if buyer \(i\) is active
and a low-priced benchmark item \(j\in O_i\) survives until \(i\), then \(j\) is available
to \(i\) at price at most \(w_j/2\).  Since \(u'_i=u_i\) for every nonzero benchmark
buyer, the supporting XOS clause lets buyer \(i\) aggregate all such surviving low-priced
benchmark items and obtain utility at least half their total supporting contribution.
Revenue is counted once per sold item, and each utility term is attached to the unique
benchmark owner.

By \Cref{lem:role-exposure}(b),(c), using
\(A^{\rm low}_{j,t}\le S^{\rm low}_{j,t}\le1\), for each \(j,t\),
\[
\begin{aligned}
  \E\!\left[
    S^{\rm high}_{j,t}
    +\1\{\ell_{i(j)}=t\}(\Lambda_{j,t}-A^{\rm low}_{j,t})
  \right] \quad\ge \quad
  \frac1{K+1}
  \E\!\left[
    S^{\rm high}_{j,t}
    +\Lambda_{j,t}
    -S^{\rm low}_{j,t}
    -O(H/K)
  \right].
\end{aligned}
\]
Average over \(t=1,\ldots,H\).  For the cut \(b(j)\), the empirical cut statistic gives
\[
  \overline S^{\rm high}_j+\overline\Lambda_j-\overline S^{\rm low}_j
  \ge
  \frac12-\widehat\Delta_j ,
\]
where bars denote averages over the \(H\) training rounds.  By
\Cref{thm:cut-regret,lem:cut-calibration},
\[
  \E[\widehat\Delta_j]
  \le
  \alpha+O\!\left(\sqrt{\frac{\log(2N)}{H}}\right).
\]
With constant perturbation parameters, \(\alpha\) is a sufficiently small constant.  Taking
\(H=C\log(2N)\) with \(C\) large, and then \(K=C_KH\) with \(C_K\) large, makes the bracket
above bounded below by a positive absolute constant.  Substituting into
\eqref{eq:role-item-cert} and summing over midrange items proves the lemma.
\end{proof}

\subsubsection*{Chronological blocks}

We next prove the random-order analogue.

\begin{lemma}[Block-prefix exposure]\label{lem:block-prefix-exposure}
Condition on the \(\mathrm{STAT}\) block and on the set of real buyers remaining after
\(\mathrm{STAT}\).  Label the dummy buyers arbitrarily for the analysis, and let every
one of the \(K+1\) blocks have size \(s\).

Fix a real post-\(\mathrm{STAT}\) buyer \(i\), a training round \(t\le H\), and let
\(\mathcal H_t\) be the history immediately before block \(t\) is exposed.  In particular,
the stored price vector \(P_t\) is \(\mathcal H_t\)-measurable.  Let \(F_i\) be the event
that \(i\) lies in the final allocation block.

Let \(Z^{\rm alloc}_t\in[0,1]\) be a statistic computed from \(P_t\) and the ordered
prefix preceding \(i\) in the allocation block.  Let \(Z^{\rm tr,del}_t\in[0,1]\) be the
corresponding statistic computed from \(P_t\) and the first
\[
  R\sim\Unif\{0,\ldots,s-1\}
\]
positions of training block \(t\), after deleting \(i\) if she occurs in that prefix.
Here \(R\) is drawn independently of the block permutation and of the mechanism's past
randomness.  Formally, the two statistics use the same measurable function of the price
vector and ordered prefix, together with auxiliary randomness of the same law that is
independent of the block permutation.
The common statistic may also depend on any $\mathcal H_t$-measurable quantities, such as the current learner distributions.

Then
\[
  \E[\1_{F_i}Z^{\rm alloc}_t]
  =
  \frac1{K+1}\E[Z^{\rm tr,del}_t]
  \pm
  O\!\left(\frac1{K+1}\cdot\frac HK\right).
\]
The same estimate holds for the actual training-prefix statistic
\(Z^{\rm tr}_t\), without deleting \(i\).

Moreover, if \(X_t\in[0,1]\) is \(\mathcal H_t\)-measurable, then
\[
  \E[\1_{F_i}X_t]
  =
  \frac1{K+1}\E[X_t]
  \pm
  O\!\left(\frac1{K+1}\cdot\frac HK\right).
\]
\end{lemma}

\begin{proof}
After labeling the dummies, the post-\(\mathrm{STAT}\) sequence is a uniformly random
permutation of all remaining real and dummy buyers.  Let
\[
  E_t:=\{i\text{ does not occur in blocks }1,\ldots,t-1\},
  \qquad
  q_t:=K+2-t.
\]
The event \(E_t\) is \(\mathcal H_t\)-measurable.  Conditional on
\(\mathcal H_t\) and \(E_t\), the still-unexposed sequence is uniformly random, and the
block containing \(i\) is uniform among the \(q_t\) remaining blocks.  Consequently,
\begin{equation}\label{eq:block-probability}
  \Prb[F_i\mid\mathcal H_t,E_t]=\frac1{q_t}.
\end{equation}
Also,
  $\Prb[E_t^c]=\frac{t-1}{K+1}\le\frac{H}{K+1}$,
and, since \(t\le H\) and \(K\ge2H\),
\begin{equation}\label{eq:block-probability-approx}
  \left|\frac1{q_t}-\frac1{K+1}\right|
  =
  \frac{t-1}{q_t(K+1)}
  =
  O\!\left(\frac{H}{K(K+1)}\right).
\end{equation}

We first compare the two prefix distributions.  Conditional on
\(\mathcal H_t,E_t,F_i\), the position of \(i\) in the allocation block is uniform in
\(\{1,\ldots,s\}\).  Thus the number of positions preceding \(i\) is uniform in
\(\{0,\ldots,s-1\}\).  Conditional on this number being \(R\), the ordered prefix before
\(i\) is a uniformly random ordered \(R\)-tuple from the still-unexposed buyers other
than \(i\).

Define an \emph{ideal prefix} by independently drawing
\(R\sim\Unif\{0,\ldots,s-1\}\) and then drawing a uniformly random ordered \(R\)-tuple
from the still-unexposed buyers other than \(i\).  By the preceding paragraph, this is
exactly the conditional distribution of the prefix preceding \(i\) in the allocation
block.

Now consider the first \(R\) positions of training block \(t\).  Conditional on
\(\mathcal H_t,E_t\), buyer \(i\) is uniform among the \(q_ts\) still-unexposed
positions.  Therefore
\[
  \Prb[i\text{ belongs to the selected training prefix}
       \mid\mathcal H_t,E_t]
  =
  \frac{\E[R]}{q_ts}
  =
  \frac{s-1}{2q_ts}
  \le
  \frac1{2q_t}.
\]
For the coupling, first condition on $R=r$. The selected training prefix is a uniformly random ordered $r$-tuple from all $q_ts$ unexposed buyers. Couple it to a uniformly random ordered $r$-tuple excluding $i$ by keeping it unchanged when it does not contain $i$, and otherwise replacing $i$ by a uniformly random buyer outside the tuple. Thus the deleted training prefix and the ideal prefix can differ only when the original training prefix contains $i$. Averaging over the original uniform choice of $R$, their mismatch probability is at most $\E[R]/(q_ts)\le1/(2q_t)$.

Let \(A_t\in[0,1]\) denote the conditional expectation of the statistic evaluated on the
ideal prefix, given \(\mathcal H_t,E_t\).  The allocation-prefix identity above and
\eqref{eq:block-probability} give
\[
  \E[\1_{F_i}Z^{\rm alloc}_t]
  =
  \frac1{q_t}\E[\1_{E_t}A_t].
\]
The prefix coupling gives
\[
  \left|
    \E[\1_{E_t}A_t]-\E[Z^{\rm tr,del}_t]
  \right|
  \le
  \frac1{2q_t}+\Prb[E_t^c].
\]
Indeed, on \(E_t\) the conditional expectations differ by at most \(1/(2q_t)\), while
on \(E_t^c\) the statistic is merely bounded by one.

Writing \(c_0:=1/(K+1)\), we conclude that
\[
  \left|
    \E[\1_{F_i}Z^{\rm alloc}_t]
    -
    c_0\E[Z^{\rm tr,del}_t]
  \right|
  ~\le~
  \left|\frac1{q_t}-c_0\right|
  +
  c_0\left(\frac1{2q_t}+\Prb[E_t^c]\right)
  ~=~
  O\!\left(\frac1{K+1}\cdot\frac HK\right),
\]
where we used \eqref{eq:block-probability-approx}, \(q_t=\Omega(K)\), and
\(\Prb[E_t^c]\le H/(K+1)\).

The actual and deleted training prefixes differ only if the selected prefix contains
\(i\), whose probability conditional on \(E_t\) is at most \(1/(2q_t)\).  Since the
statistic is in \([0,1]\), replacing \(Z^{\rm tr,del}_t\) by \(Z^{\rm tr}_t\) introduces
an additional
 $ O\!\left(\frac1{K+1}\cdot\frac1K\right)$
error, which is absorbed in the displayed bound.

Finally, suppose \(X_t\) is \(\mathcal H_t\)-measurable.  By
\eqref{eq:block-probability}, we have 
  $\E[\1_{F_i}X_t]
  =
  \frac1{q_t}\E[\1_{E_t}X_t]$.
Therefore,
\[
  \left|
    \E[\1_{F_i}X_t]-\frac1{K+1}\E[X_t]
  \right|
  ~\le~
  \left|\frac1{q_t}-\frac1{K+1}\right|
  +
  \frac1{K+1}\Prb[E_t^c]
  ~=~
  O\!\left(\frac1{K+1}\cdot\frac HK\right). \qedhere
\]
\end{proof}

\begin{lemma}[Midrange capture for chronological blocks]
\label{lem:midrange-blocks}
In the chronological-block implementation,
\[
  \E[\ALG_{\rm sparse}]
  \ge
  \frac{c}{K+1}\,\mathrm{Mid}_\Gamma(u')
\]
for an absolute constant \(c>0\), where the expectation is over the random order and the
sparse-branch randomness, conditioned on the STAT block and the scale \(\Gamma\).
\end{lemma}

\begin{proof}
Fix the benchmark allocation of \(u'\) and its supporting clauses.  Dummies have zero
valuation and are never benchmark owners.

Consider a midrange item \(j\), write \(i=i(j)\) and \(w=w_j\), and let \(b=b(j)\) be the
largest grid index such that \(p_be^\tau\le w/2\).  In round \(t\), let \(Y_{j,t}(r)\) be
the smoothed sold indicator of item \(j\) on the random training prefix when item \(j\)'s
base level is \(r\), holding \(P_{-j,t}\) fixed.  Define
\[
  \Lambda_{j,t}:=\sum_{r\le b}\pi_{j,t}(r),
  \qquad
  S^{\rm low}_{j,t}:=\sum_{r\le b}\pi_{j,t}(r)Y_{j,t}(r),
  \qquad
  S^{\rm high}_{j,t}:=\sum_{r>b}\pi_{j,t}(r)Y_{j,t}(r).
\]
For the allocation block, define \(S^{A,{\rm low}}_{j,t}\) and
\(S^{A,{\rm high}}_{j,t}\) analogously, using the prefix before buyer \(i\) in the final
allocation block.

If \(i\) lies in the allocation block, high-price sales before \(i\) give revenue
\(\Omega(w)\).  If a low-priced benchmark item has not sold before \(i\), then it is
available to \(i\) at price at most \(w/2\), and the supporting XOS clause gives utility
\(\Omega(w)\) from this item, aggregating over all benchmark items of the same buyer.  As
in the role proof,
\[
  \E[\text{round-\(t\) welfare contribution of item }j]
  \ge
  \kappa w\,
  \E\!\left[
    \1\{i\text{ in allocation block}\}
    (S^{A,{\rm high}}_{j,t}+\Lambda_{j,t}-S^{A,{\rm low}}_{j,t})
  \right]
\]
for an absolute constant \(\kappa>0\).

Apply \Cref{lem:block-prefix-exposure}: use the lower comparison for the high-sale prefix
statistic, the predictable comparison for \(\Lambda_{j,t}\), and the upper comparison for
the low-sale prefix statistic.  Since all statistics lie in \([0,1]\), the expected
contribution of item \(j\) in round \(t\) is at least
\[
  \frac{\kappa w}{K+1}
  \E\!\left[
    S^{\rm high}_{j,t}+\Lambda_{j,t}-S^{\rm low}_{j,t}-O(H/K)
  \right].
\]
Average over the uniformly chosen stored round.  For the cut \(b(j)\), the same algebra as
in the role proof gives
\[
  \overline S^{\rm high}_j+\overline\Lambda_j-\overline S^{\rm low}_j
  \ge
  \frac12-\widehat\Delta_j .
\]
The path cut-regret theorem and \Cref{lem:cut-calibration} make
\(\E[\widehat\Delta_j]\) a sufficiently small constant by the choice
\(H=C\log(2N)\) and small constant perturbation parameters.  Choosing \(C_K\) large makes
\(H/K\) sufficiently small.  Thus every midrange item contributes
\(\Omega(w_j/(K+1))\) in expectation.  Summing over midrange items proves the lemma.
\end{proof}

\begin{proof}[Proof of \Cref{lem:sparse-scale}]
For the zero-padded role implementation, \Cref{lem:midrange-roles} gives
\[
  \E[\ALG_{\rm sparse}]
  \ge
  \frac{c}{K+1}\,\mathrm{Mid}_\Gamma(u').
\]
For the chronological-block implementation, \Cref{lem:midrange-blocks} gives the same
bound.  In both cases, \eqref{eq:mid-large-enough} implies
\[
  \mathrm{Mid}_\Gamma(u')\ge \OPT(u')-O(\rho\Gamma).
\]
Absorbing constants proves \Cref{lem:sparse-scale}.
\end{proof}

{\small
\bibliographystyle{alpha}
\bibliography{bib}
}
\end{document}